\documentclass[12pt,a4paper]{article}
\usepackage{amsfonts}
\usepackage{graphicx}
\usepackage{indentfirst}
\usepackage{amsmath,amssymb,amsthm}
\usepackage{amssymb}
\usepackage{algorithm,algpseudocode}
\usepackage{latexsym}
\usepackage{natbib}
\usepackage[margin=1in,a4paper]{geometry}
\usepackage{colortbl}
\usepackage{color}
\usepackage[colorlinks,citecolor=blue,urlcolor=blue]{hyperref}
\usepackage{mathtools}
\usepackage{setspace}
\usepackage{verbatim}
\usepackage{subfigure}
\usepackage[title]{appendix}

\mathtoolsset{showonlyrefs=true}

\newtheorem{theorem}{Theorem}
\theoremstyle{remark}
\newtheorem{remark}{Remark}

\allowdisplaybreaks[4]

\graphicspath{{../../Codes/}}

\newcommand*\diff{\mathop{}\!\mathrm{d}}

\title{Stochastic Volatility Model with Sticky Drawdown and Drawup Processes: A Deep Learning Approach}

\author{Yuhao Liu \thanks{School of Science and Engineering, The Chinese University of Hong Kong, Shenzhen, China. Email: oxhowardliu@outlook.com} \and Pingping Jiang \thanks{Center for Financial Engineering, Soochow University, Suzhou, China. Email: ppjiang@suda.edu.cn} \and Gongqiu Zhang\thanks{School of Science and Engineering, The Chinese University of Hong Kong, Shenzhen, China. Email: zhanggongqiu@cuhk.edu.cn.}}

\begin{document}

\maketitle

\begin{abstract}
    We propose a new financial model, the stochastic volatility model with sticky drawdown and drawup processes (SVSDU model), which enables us to capture the features of winning and losing streaks that are common across financial markets but can not be captured simultaneously by the existing financial models. Moreover, the SVSDU model retains the advantages of the stochastic volatility models. Since there are not closed-form option pricing formulas under the SVSDU model and the existing simulation methods for the sticky diffusion processes are really time-consuming, we develop a deep neural network to solve the corresponding high-dimensional parametric partial differential equation (PDE), where the solution to the PDE is the pricing function of a European option according to the Feynman–Kac Theorem, and validate the accuracy and efficiency of our deep learning approach. We also propose a novel calibration framework for our model, and demonstrate the calibration performances of our models on both simulated data and historical data. The calibration results on SPX option data show that the SVSDU model is a good representation of the asset value dynamic, and both winning and losing streaks are accounted for in option values. Our model opens new horizons for modeling and predicting the dynamics of asset prices in financial markets.

\end{abstract}

\section{Introduction}
Persistent extremes of asset values appear frequently in financial markets. In the bull market, asset prices continue to rise, leading to more buying and thereby further driving up prices. While in the bear market, the continuous decline in asset prices leads to less demand for the assets, which further pushes down the prices. So the record highs or lows of asset prices tend to be clustered for concentrated periods of time. Although persistent extremes are pervasive in financial markets, most of the financial models such as the Black-Scholes model and the Heston model can not capture this notable feature, since the amount of time that the underlying process spends in its running maximum and running minimum always has measure zero in those models. To incorporate this feature, \cite{feng2020modeling} proposed a new financial model driven by a sticky maximum process, which can explain the feature of winning streaks that refer to the phenomenon in which asset values consistently increase over an uninterrupted period. Because their model is based on geometric Brownian motion, option prices are solvable analytically under the model. However, since the model in \cite{feng2020modeling} assumes that the volatility is constant, it does not account for other stylized facts of financial data such as the volatility smile. Besides, their  model does not consider the feature of losing streaks, which usually appear in economic downturns. Motivated by \cite{feng2020modeling}, we develop a new stochastic volatility model driven by sticky diffusion processes for derivatives pricing, which not only captures the features of winning and losing streaks in the financial market, but also allows the volatility to vary over time. We call our model stochastic volatility model with sticky drawdown and drawup processes (SVSDU model).

Explorations for sticky diffusion processes stemmed from \cite{feller1952parabolic}. Since then, sticky behavior of stochastic processes has been widely studied by academy, including \cite{harrison1981sticky}, \cite{graham1988martingale},  \cite{bass2014stochastic}, \cite{engelbert2014stochastic}, \cite{racz2015multidimensional}, \cite{grothaus2017stochastic}, \cite{salins2017markov}, etc. However, there have been relatively few applications of sticky processes in finance. Existing literature is limited on stochastic processes exhibiting sticky reflecting behavior in one dimension for financial applications. \cite{jiang2019some} modeled the price clustering effect by constructing a sticky diffusion process with sticky reflection at a fixed point and computed option values based on the sticky process. \cite{nie2020sticky} applied sticky reflecting Ornstein-Uhlenbeck diffusions to model the interest rate that follows a sticky reflecting behavior at zero. \cite{feng2020modeling} proposed the sticky drawdown process with sticky reflection at 1 to model the winning streak. The existence and uniqueness of the solution to the stochastic differential equation (SDE) for the sticky diffusions are studied in  \cite{graham1988martingale}, \cite{takanobu1988existence} and \cite{ikeda2014stochastic}. \cite{graham1988martingale} proved that there exists at least one solution to the multidimensional SDE system with more than one dimension exhibiting stickiness if sojourn without reflection is allowed at the boundary. 

In order to model both the winning and losing streaks, we apply the technique of time change in \cite{salins2017markov} to the drawdown and drawup processes to construct sticky diffusions, which exhibit stickiness in two dimensions. As a result, the asset dynamics in our model can stay at their running maximums and running minimums for a positive period of time. We derive a multidimensional SDE system for our sticky diffusions, where the volatility dynamic follows Cox–Ingersoll–Ross process. We use the technique of change of variables to prove the existence of the solutions to our SDE system with the theorem in \cite{graham1988martingale}.

Although the SVSDU model has desirable properties, it gives rise to nontrivial mathematical issues: it is quite challenging to obtain closed form solutions to no arbitrage option prices. To address this problem, we propose an unsupervised deep learning approach: we derive a partial differential equation (PDE) whose solution is the pricing function of a European option under the SVSDU model, and approximate the solution to the PDE by a deep neural network. There has been rapid developments in application of deep learning approaches in solving PDEs, including physics-informed neural networks (\cite{raissi2019physics}), Weak adversarial networks (\cite{zang2020weak}), random feature method (\cite{chen2022bridging}), and deep Galerkin method (\cite{sirignano2018dgm}). For application in financial PDEs such as the Black-Scholes PDE, see \cite{sirignano2018dgm} and \cite{glau2022deep}. Our neural network is designed based on the framework proposed by \cite{sirignano2018dgm}, which has demonstrated its approximation power in solving high-dimensional PDEs. This approach is efficient because the neural network generates option prices for a wide range of parameters within milliseconds after training. To assess the accuracy of our deep learning approach, we calculate the benchmarks of option prices by adapting the \cite{meier2023simulation} simulation algorithm to simulate the diffusion processes with sticky boundaries. Using the results from Monte Carlo method with a large number of paths, we show that the option prices generated by the deep neural network match the benchmarks.

We also calibrate parameters by fitting the model to option data, as option pricing formula is approximated by the deep neural network. There has been some work on neural network calibration for different financial models including \cite{horvath2021deep}, \cite{liu2019neural}, \cite{bayer2019deep} and \cite{romer2022empirical}. In this paper, we propose a calibration framework for the SVSDU model and show the calibration accuracy and efficiency on simulated data. The main challenge in the calibration is that since the input variables of the neural network such as asset price variables and strike price variables are within bounded domains when we train the neural network with the PDE's residuals as loss, the calibration performances may suffer if the magnitudes of the input data lie significantly outside the input domain of the network. So we develop an efficient approach to find a suitable scaling factor to standardize the input data if the input values are too large for the neural network during the calibration.

In empirical studies, we perform calibration tasks to SPX options in 2021 and 2022. To analyze the effect of individual stickiness factor and the joint effect of drawdown and drawup stickiness factors, we consider two new stochastic models: stochastic volatility model with sticky drawdown processes (SVSD model) that only considers the feature of winning streaks and stochastic volatility model with sticky drawup processes (SVSU model) that only considers the feature of losing streaks. Option pricing functions under the SVSD model and the SVSU model are approximated by the other two neural networks. We analyze the performances of the SVSDU model, SVSD model and the SVSU model in different market situations and compare them with the Heston model. The numerical result shows that  the SVSDU model, the SVSD model and the SVSU model perform better than the Heston model across all market scenarios both in-sample and out-of-sample. Moreover, the SVSDU model has the best fitting performance in the whole calibration period and the best prediction performance in the period when both the winning and losing streaks appear, while the SVSD model (SVSU model) achieves good prediction performances when the index values continue to rise (fall) in a period of time. Therefore, both the winning and losing streaks are accounted for in option values, and the SVSDU model is a good reflection of market data due to the fact that the asset values will not keep rising continuously, nor will they decline indefinitely.

The rest of this paper is organized as follows. In Section \ref{sec-modelsimu}, we introduce the SVSDU model, provide the simulation scheme for the model, and derive the corresponding pricing PDE. In Section \ref{sec-deep}, we show how to train a deep neural network to approximate the option pricing function under the SVSDU model and provide a calibration method for the model. Section \ref{sec-num} presents the numerical results of our pricing and calibration approach on simulated data. Section \ref{sec-emp} shows the calibration results of the SVSDU model on SPX option data in different market situations and compares it to three other models. Section \ref{sec-con} concludes this paper. All the proofs are collected in Appendix \ref{app-proof}. The detailed calibration algorithm is presented in Appendix \ref{app-cali}, and the SVSD model and SVSU model are shown in Appendix \ref{app-SVSDSVSU}.

\section{The Stochastic Volatility Model with Sticky Drawdown and Drawup processes}
\label{sec-modelsimu}
\subsection{The Model Setup}
\label{sec-SVSDU}
We start from the Black-Scholes model under a physical probability measure $\mathbf{P}$, i.e.
\[
d S_{t}=\mu S_{t} d t+\hat{\sigma} S_{t} d \widetilde{B}_{t},
\]
where $\mu\in \mathbb{R}$, $\hat{\sigma}\in \mathbb{R}_+$ and $\{\widetilde{B}_t, t\ge 0\}$ is a standard Brownian motion under $\mathbf{P}$. 
We consider the risk-adjusted asset value process, i.e. $\widetilde{S}_t=e^{-rt}S_t$, where $r$ is the yield rate of risk-free asset, to remove the long-term uptrend effect. The risk-adjusted value process satisfies
\[
	d \widetilde{S}_{t}=(\mu-r) \widetilde{S}_{t} d t+\hat{\sigma} \widetilde{S}_{t} d \widetilde{B}_{t} 
\]
The running maximum and running minimum of risk-adjusted asset value  $\widetilde{S}_t$ up to time $t$ are given by $\{\overline{S}_{t}, t\ge 0\}$ and $\{\underline{S}_{t}, t\ge 0\}$, where
\begin{align}
\overline{S}_{t}=\sup_{0\leq u\leq t} \widetilde{S}_u, \ \underline{S}_{t}=\inf_{0\leq u\leq t} \widetilde{S}_u.
\end{align}
So the drawdown process $D_{t}$ and drawup process $U_{t}$ can be defined as:
\begin{align}
D_{t}=\frac{\widetilde{S}_{t}}{\overline{S}_{t}}, \ U_{t}=\frac{\widetilde{S}_{t}}{\underline{S}_{t}}.
\end{align}
If $D_t=1$, the asset value achieves its running maximum and reports a record high. If $U_t=1$, a new record-breaking low appears in the asset value. According to \cite{feng2020modeling}, we can easily derive the SDE for $D_t$ and $U_t$:
\begin{align}\label{eq:drawdown and drawup process}
    \begin{cases}
        &d D_t=(\mu-r)D_tdt+\hat{\sigma} D_td\widetilde{B}_t-D_t dL_t^1(D)\\
        &d U_t=(\mu-r)U_tdt+\hat{\sigma} U_td\widetilde{B}_t+U_tdL_t^1(U)
    \end{cases}
\end{align}
where $L_t^1(D)$ is the local time of the process $D$ at 1 and $L_t^1(U)$ is the local time of the process $U$ at 1. 1 represents the upper reflecting barrier for $D$ and the lower reflecting barrier for $U$. The process $\{L_t, t\ge 0 \}$ only increases when $D$ reaches the upper reflecting barrier or $U$ touches the lower reflecting barrier, so $D$ and $U$ will be prevented from moving above and falling below 1 when they hit the barriers.

However, the processes defined in \eqref{eq:drawdown and drawup process} can not model the persistence of winning and losing streaks, because the occupation time of a process driven by a Brownian motion in a fixed point is known to be zero according to \cite{salins2017markov}. In order to address this problem,  we apply the technique of time change as described in \cite{feng2020modeling} to derive stochastic differential equations for the sticky drawdown process, sticky drawup process and the corresponding asset value process. We consider a random clock $R(t)$ defined as:
\begin{align}
R(t) = t + \xi L_t^1(D) + \eta L_t^1(U),
\end{align}
where $\xi >0$ is the drawdown stickiness coefficient and $\eta>0$ is the drawup stickiness coefficient.  Then the sticky drawdown process $D^\pm_t$ and sticky drawup process $U^\pm_t$ can be obtained by time-changing the original processes $D_t$ and $U_t$, i.e.
\begin{align}
D^\pm_t = D_{R^{-1}(t)},\ U^\pm_t = U_{R^{-1}(t)},
\end{align}
where $R^{-1}(t)=t-\xi L_{t}^{1}\left(D^{ \pm}\right)-\eta L_{t}^{1}\left(U^{ \pm}\right)$. Compared with the original processes, the sticky drawdown process $D^\pm_t$ and sticky drawup process $U^\pm_t$ spend a positive amount of time at 1 with similar arguments in \cite{feng2020modeling}. The corresponding risk-adjusted sticky value process $\widetilde{S}^\pm_t$ can also be obtained by time-changing the processes $\widetilde{S}_t$, i.e. 
\begin{align}
    \widetilde{S}^\pm_t = \widetilde{S}_{R^{-1}(t)}.
\end{align}
The running maximum and running minimum of the risk-adjusted sticky value process are
\[
\overline{M}^\pm_t=\sup_{0\le u\le t}\widetilde{S}^\pm_u,\ \underline{M}^\pm_t=\inf_{0\le u\le t}\widetilde{S}^\pm_u.
\]
Different from the original asset value process, $\widetilde{S}^\pm_t$ can stay in the running maximum or the running minimum for a long stretch of time. So the new process exhibits the phenomena of winning and losing streaks over a long period. Similar with \cite{feng2020modeling}, we can derive the SDE system for $D_t^\pm$ and $U_t^\pm$ under the constant volatility:
\begin{align}\label{eq:sticky-process-SDE}
	\begin{cases}
	&dD^\pm_t = \mathbf{1}_{\{ D^\pm_t \ne 1, U^\pm_t \ne 1 \}}\big((\mu-r) D^\pm_t dt + \hat{\sigma} D^\pm_t dB_t \big)  - D^\pm_t dL_t^1(D^\pm),\\
	&dU^\pm_t = \mathbf{1}_{\{ D^\pm_t \ne 1, U^\pm_t \ne 1 \}}\big((\mu-r) U^\pm_t dt + \hat{\sigma} U^\pm_t dB_t \big)  + U^\pm_t d L_t^1(U^\pm),\\
	&\int_{0}^{t} \mathbf{1}_{\{D^\pm_s = 1\}} ds  =  \xi L_t^1(D^\pm_t),\\
	&\int_{0}^{t} \mathbf{1}_{\{U^\pm_s = 1\}} ds  =  \eta L_t^1(U^\pm_t).
	\end{cases}	
	\end{align}
where $B_t=\widetilde{B}_{R^{-1}(t)}+\hat{B}_{t-R^{-1}(t)}$ and the Brownian motion $\{\hat{B}_t, t\ge 0\}$ is independent of $\{\widetilde{B}_t, t\ge 0\}$.

Empirical evidence from financial markets clearly shows that the volatility is not constant. Instead, it varies randomly in time. So under the physical measure $\mathbf{P}$, we propose the following stochastic volatility model with sticky drawdown and drawup process:
\begin{align}\label{eq:sticky-process-SDE}
	\begin{cases}
	&dD^\pm_t = \mathbf{1}_{\{ D^\pm_t \ne 1, U^\pm_t \ne 1 \}}\big((\mu-r) D^\pm_t dt + \sqrt{V_t} D^\pm_t dB_t \big)  - D^\pm_t dL_t^1(D^\pm),\\
	&dU^\pm_t = \mathbf{1}_{\{ D^\pm_t \ne 1, U^\pm_t \ne 1 \}}\big((\mu-r) U^\pm_t dt + \sqrt{V_t} U^\pm_t dB_t \big)  + U^\pm_t d L_t^1(U^\pm),\\
	&dV_t=\kappa(\theta-V_t)dt+\sigma\sqrt{V_t}dW_t,\\
	&\int_{0}^{t} \mathbf{1}_{\{D^\pm_s = 1\}} ds  =  \xi L_t^1(D^\pm_t),\\
	&\int_{0}^{t} \mathbf{1}_{\{U^\pm_s = 1\}} ds  =  \eta L_t^1(U^\pm_t).
	\end{cases}	
	\end{align}
 where $\sigma\in \mathbb{R}_+$, $\{B_t, t\ge 0\}$ and $\{W_t, t\ge 0\}$ are two standard Brownian motion under measure $\mathbb{P}$ with the correlation $\rho$. The variance $V_t$ follows CIR process, and the Feller condition $2\kappa\theta>\sigma^2$ is satisfied. The process $D^{\pm}_t$ and $U^{\pm}_t$ exhibit stickiness at one, while $V_t$ is not sticky and $V_t>0$. 
 \begin{theorem}\label{thm:existence-weak-sol-SDE}
     There exists at least one solution to the SDE system \eqref{eq:sticky-process-SDE}.
 \end{theorem}
 The proof of the Theorem \ref{thm:existence-weak-sol-SDE} is shown in Appendix \ref{app-proof}.
 
 Similarly, the dynamic of the risk-adjusted sticky value process with both sticky running maximum and sticky running minimum under the physical measure is given by 
 \begin{align}
	\left\{\begin{array}{l}
	d \widetilde{S}_t^{\pm}= (\mu-r) \widetilde{S}_t^{\pm} \mathbf{1}_{\left\{\widetilde{S}_t^{\pm} \neq\left\{\overline{M}_t^{\pm}, \underline{M}_t^{\pm}\right\}\right\}} d t+\sqrt{V_t} \widetilde{S}_t^{\pm} \mathbf{1}_{\left\{\widetilde{S}_t^{\pm} \neq\left\{\overline{M}_t^{\pm}, \underline{M}_t^{\pm}\right\}\right\}} d B_t, \\
	dV_t=\kappa(\theta-V_t)dt+\sigma\sqrt{V_t}dW_t \\
	\int_0^t \mathbf{1}_{\left\{\widetilde{S}_t^{\pm}=\overline{M}_t^{\pm}\right\}} d s=\xi L_t^1\left(D^{\pm}\right) \\
	\int_0^t \mathbf{1}_{\left\{\widetilde{S}_t^{\pm}=\underline{M}_t^{\pm}\right\}} d s=\eta L_t^1\left(U^{\pm}\right) .
	\end{array}\right.
	\end{align}
 Since $D^\pm_t=\frac{\widetilde{S}_t^{\pm}}{\overline{M}^\pm_t}$ and $U^\pm_t=\frac{\widetilde{S}_t^{\pm}}{\underline{M}^\pm_t}$, we can derive the SDE for sticky running maximum process and sticky running minimum process:
 \begin{align}
 \left\{\begin{array}{l}
	d\overline{M}^\pm_t = \overline{M}^\pm_t dL_t^1(D^\pm),\\
	d\underline{M}^\pm_t = -\underline{M}^\pm_t dL_t^1(U^\pm).
 \end{array}\right.
\end{align}
Consequently, we can derive the SDE for the sticky asset value process $S_t^{\pm}$ defined by $S_t^{\pm}=e^{rt}\widetilde{S}^+_t$ under the physical measure:
\begin{align}
	\left\{\begin{array}{l}
	d S_t^{\pm}=r{S}_t^{\pm}dt+ (\mu-r) S_t^{\pm} \mathbf{1}_{\left\{{S}_t^{\pm} \neq\left\{\overline{S}_t^{\pm}, \underline{S}_t^{\pm}\right\}\right\}} d t+\sqrt{V_t} {S}_t^{\pm} \mathbf{1}_{\left\{{S}_t^{\pm} \neq\left\{\overline{S}_t^{\pm}, \underline{S}_t^{\pm}\right\}\right\}} d B_t, \\
	dV_t=\kappa(\theta-V_t)dt+\sigma\sqrt{V_t}dW_t \\
	\int_0^t \mathbf{1}_{\left\{{S}_t^{\pm}=\overline{S}_t^{\pm}\right\}} d s=\xi L_t^1\left(D^{\pm}\right) \\
	\int_0^t \mathbf{1}_{\left\{{S}_t^{\pm}=\underline{S}_t^{\pm}\right\}} d s=\eta L_t^1\left(U^{\pm}\right) ,
	\end{array}\right.
	\end{align}
where $\overline{S}_t^{\pm}=e^{rt}\overline{M}_t^{\pm}$ and $\underline{S}_t^{\pm}=e^{rt}\underline{M}_t^{\pm}$. It is easy to derive that 
\begin{align}
\left\{\begin{array}{l}
d \overline{S}_t^{ \pm}=r \overline{S}_t^{ \pm} d t+\overline{S}_t^{ \pm} d L_t^1\left(D^{ \pm}\right) \\
d \underline{S}_t^{ \pm}=r \underline{S}_t^{ \pm} d t-\underline{S}_t^{ \pm} d L_t^1\left(U^{ \pm}\right) .
\end{array}\right.
\end{align}

Under the risk-neutral measure $\mathbb{Q}$, the sticky drawdown process and sticky drawup process satisfy the following SDE:
\begin{align}
\left\{\begin{array}{l}d D_t^{ \pm}=\mathbf{1}_{\left\{D_t^{ \pm} \neq 1, U_t^{ \pm} \neq 1\right\}} \sqrt{V_t} D_t^{ \pm} d B^{\mathbb{Q}}_t-D_t^{ \pm} d L_t^1\left(D^{ \pm}\right), \\ d U_t^{ \pm}=\mathbf{1}_{\left\{D_t^{ \pm} \neq 1, U_t^{ \pm} \neq 1\right\}} \sqrt{V_t} U_t^{ \pm} d B^{\mathbb{Q}}_t+U_t^{ \pm} d L_t^1\left(U^{ \pm}\right), \\ dV_t=\kappa(\theta-V_t)dt+\sigma\sqrt{V_t}dW^{\mathbb{Q}}_t , \\  \int_0^t \mathbf{1}_{\left\{D_s^{ \pm}=1\right\}} d s=\xi L_t^1\left(D^{ \pm}\right), \\ \int_0^t \mathbf{1}_{\left\{U_s^{ \pm}=1\right\}} d s=\eta L_t^1\left(U^{ \pm}\right) ,\end{array}\right.
\end{align}
where $\{B^{\mathbb{Q}}_t, t\ge 0\}$ and $\{W^{\mathbb{Q}}_t, t\ge 0\}$ are two standard Brownian motion under measure $\mathbb{Q}$. Moreover, the sticky asset value process satisfy
\begin{align}
& \left\{\begin{array}{l}
d S_t^{ \pm}=r S_t^{ \pm} d t+\sqrt{V_t} S_t^{ \pm} \mathbf{1}_{\left\{S_t^{ \pm} \neq\left\{\overline{S}_t^{ \pm}, \underline{S}_t^{ \pm}\right\}\right\}} d B^{\mathbb{Q}}_t, \\
dV_t=\kappa(\theta-V_t^\pm)dt+\sigma\sqrt{V_t^\pm}dW^{\mathbb{Q}}_t \\
\int_0^t \mathbf{1}_{\left\{S_t^{ \pm}=\overline{S}_t^{ \pm}\right\}} d s=\xi L_t^1\left(D^{ \pm}\right) \\
\int_0^t \mathbf{1}_{\left\{S_t^{ \pm}=\underline{S}_t^{ \pm}\right\}} d s=\eta L_t^1\left(U^{ \pm}\right) .
\end{array}\right. \\
&
\end{align}
It follows that
\begin{align}
\left\{\begin{array}{l}
d \overline{S}_t^{ \pm}=r \overline{S}_t^{ \pm} d t+\overline{S}_t^{ \pm} d L_t^1\left(D^{ \pm}\right) \\
d \underline{S}_t^{ \pm}=r \underline{S}_t^{ \pm} d t-\underline{S}_t^{ \pm} d L_t^1\left(U^{ \pm}\right) .
\end{array}\right.
\end{align}

As a result of the time change, the asset value process has the desired property of persistent extremes, where the sets of periods of record highs and lows both have positive measures. Such a process is suitable for modeling financial data with both winning streaks and losing streaks.

\subsection{Simulation of The Sticky Drawdown/Drawup Processes}
\label{sec-simu}
Option prices under the SVSDU model can be obtained by Monte Carlo simulation. Our simulation method is based on the work in \cite{meier2023simulation}.
We consider the joint simulation of $\ln D_t^{ \pm}$, $\ln U_t^{ \pm}$, $\ln \bar{S}_t^{ \pm}$ and $\ln \underline{S}_t^{ \pm}$ since the log processes are easy to simulate: 
\begin{align}\label{eq:simulation-sticky-process}
   \left\{\begin{array}{l}d \ln D_t^{ \pm}=\mathbf{1}_{\left\{D_t^{ \pm} \neq 1, U_t^{ \pm} \neq 1\right\}}\left(-\frac{1}{2} V_t d t+\sqrt{V}_t d B_t^{\mathbb{Q}}\right)-d L_t^1\left(D^{ \pm}\right), \\ d \ln U_t^{ \pm}=\mathbf{1}_{\left\{D_t^{ \pm} \neq 1, U_t^{ \pm} \neq 1\right\}}\left(-\frac{1}{2} V_t d t+\sqrt{V}_t d B_t^{\mathbb{Q}}\right)+d L_t^1\left(U^{ \pm}\right), \\ d V_t=\kappa\left(\theta-V_t\right) d t+\sigma \sqrt{V_t} d W_t,  \\ \int_0^t \mathbf{1}_{\left\{D_s^{ \pm}=1\right\}} d s=\xi L_t^1\left(D^{ \pm}\right), \\ \int_0^t \mathbf{1}_{\left\{U_s^{ \pm}=1\right\}} d s=\eta L_t^1\left(U^{ \pm}\right),\\
   d \ln \bar{S}_t^{ \pm}=r d t+d L_t^1\left(D^{ \pm}\right), \\ d \ln \underline{S}_t^{ \pm}=r d t-d L_t^1\left(U^{ \pm}\right),
   \end{array}\right. 
\end{align}

where $\ln D_0^{ \pm}=\ln \frac{\widetilde{S}^{\pm}_0}{\overline{M}_0^{\pm}}$ and $\ln U_0^{ \pm}=\ln \frac{\widetilde{S}^{\pm}_0}{\underline{M}_0^\pm}$.
The SDE \eqref{eq:simulation-sticky-process} can be reformulated as:
\begin{align*}
	d \boldsymbol{X}_t= & \boldsymbol{\mu}\left(\boldsymbol{X}_t\right) I\left(\boldsymbol{X}_t \in \mathbb{S}\right) d t+\boldsymbol{\Sigma}\left(\boldsymbol{X}_t\right) I\left(\boldsymbol{X}_t \in \mathbb{S}\right) d \mathbf{B}_{1, t} \\ & +\hat{\boldsymbol{\beta}}\left(\boldsymbol{X}_t\right) I\left(\boldsymbol{X}_t \in \partial \mathbb{S}\right) d t+\hat{\boldsymbol{\Gamma}}\left(\boldsymbol{X}_t\right) I\left(\boldsymbol{X}_t \in \partial \mathbb{S}\right) d \boldsymbol{B}_{2, t}, 
\end{align*}
where $\boldsymbol{X}_t=(\ln D_t^{ \pm}, \ln U_t^{ \pm}, V_t, \ln \bar{S}_t^{ \pm}, \ln \underline{S}_t^{ \pm})^\top \in (-\infty,0]\times [0,\infty)\times \mathbb{R} \times \mathbb{R}\times \mathbb{R}$, $\mathbf{B}_{1, t}$ and $\mathbf{B}_{2, t}$ are two independent three-dimensional Brownian motions, and  
\[
\boldsymbol{\mu}\left(\boldsymbol{X}_t\right)=\left[\begin{array}{c}\sqrt{V_t} \\ \sqrt{V_t}  \\ \sigma \sqrt{V_t} \rho \\ 0 \\ 0 \end{array}\right], \ \hat{\boldsymbol{\beta}}\left(\boldsymbol{X}_t\right)=\left[\begin{array}{l}-\frac{1}{\xi}1_{\left\{\ln D_t^{ \pm}=0\right\}} \\ \frac{1}{\eta}1_{\left\{\ln U_t^{ \pm}=0\right\}} \\ \kappa\left(\theta-V_t\right) \\ \frac{1}{\xi}1_{\left\{\ln D_t^{ \pm}=0\right\}}+r \\ -\frac{1}{\eta}1_{\left\{\ln U_t^t=0\right\}}+r\end{array}\right]
\]

\[
\boldsymbol{\Sigma}\left(\boldsymbol{X}_t\right)=\left[\begin{array}{ccc}\sqrt{V_t} & 0 & 0 \\ \sqrt{V_t} & 0 & 0 \\ \sigma \sqrt{V_t} \rho & \sigma \sqrt{V_t} \sqrt{1-\rho^2} & 0 \\ 0 & 0 & 0 \\ 0 & 0 & 0\end{array}\right], \ \hat{\boldsymbol{\Gamma}}\left(\boldsymbol{X}_t\right)=\left[\begin{array}{ccc}0 & 0 & 0 \\ 0 & 0 & 0 \\ \sigma \sqrt{V_t} & 0 & 0 \\ 0 & 0 & 0 \\ 0 & 0 & 0\end{array}\right].
\]
The diffusion lives on $\overline{\mathbb{S}}=\mathbb{S}\ \cup\ \partial \mathbb{S}$, where 
\[
\mathbb{S}=\left\{\boldsymbol{x} \in \mathbb{R}^5: \Phi(\boldsymbol{x})< 0\right\} , \partial \mathbb{S}=\left\{\boldsymbol{x} \in \mathbb{R}^5: \Phi(\boldsymbol{x})=0\right\}
\]
for some $\Phi \in C_b^2\left(\mathbb{R}^5\right)$. $\Phi$ is given by $\Phi(\boldsymbol{x})=\prod_{i=1}^{\hat{d}} (1-e^{-x_i})$, where $\hat{d}$ is the number of dimensions with stickiness and $\hat{d}=2$ (see Section 2, \cite{meier2023simulation}).

We use eigendecomposition approach proposed by \cite{meier2023simulation} to construct a CTMC for simulation of a multidimensional diffusion with sticky boundaries. Let $\boldsymbol{A}=\boldsymbol{\Sigma} \boldsymbol{\Sigma}^{\top}$ and $\hat{\boldsymbol{G}}=\hat{\boldsymbol{\Gamma}} \hat{\boldsymbol{\Gamma}}^{\top}$. They can be written as:
$$
\boldsymbol{A}=\sum_{i=1}^d \lambda_i \boldsymbol{u}_i \boldsymbol{u}_i^{\top}, \ \hat{\boldsymbol{G}}=\sum_{i=1}^d \hat{\lambda}_i \boldsymbol{u}_i^{\hat{\boldsymbol{G}}}(\boldsymbol{u}_i^{\hat{\boldsymbol{G}}})^{\top},
$$
where $\boldsymbol{u}_i$, $\boldsymbol{u}_i^{\hat{\boldsymbol{G}}}$ are the normalized eigenvectors associated with $\lambda_i$ and $\hat{\lambda}_i$. Since the drift vector and the eigenvectors give the direction along which the process varies, we use them as directions to move CTMC. Specifically, the directions are
\begin{align*}
	M(\boldsymbol{x})= \begin{cases}\{\boldsymbol{\mu}(\boldsymbol{x})\} \cup\left\{\boldsymbol{u}_i(\boldsymbol{x}),-\boldsymbol{u}_i(\boldsymbol{x}): i=1, \ldots, d\right\}, & \text { if } \boldsymbol{x} \in \mathbb{S}, \\ \{\hat{\boldsymbol{\beta}}(\boldsymbol{x})\} \cup\left\{\boldsymbol{u}_i^{\hat{\boldsymbol{G}}}(\boldsymbol{x}),-\boldsymbol{u}_i^{\hat{\boldsymbol{G}}}(\boldsymbol{x}): i=1, \ldots, d\right\}. & \text { if } \boldsymbol{x} \in \partial \mathbb{S}\end{cases}
\end{align*}

As for the step size $h$, we need to adjust it if the process $\boldsymbol{x}\in \mathbb{S}$ is close to the boundary since the original size may lead the CTMC out of the boundary. For the minimum distance of $\boldsymbol{x}$ to $\partial \mathbb{S}$ along transition direction $\boldsymbol{u}$, we have 
\[
h=\min \left\{\big|\frac{x_i}{u_i} \big|, h :  i=1, \cdots, \hat{d} \right\},
\]
where $u_i$ is the i-th element of the transition direction $\boldsymbol{u}$. The adjusted step sizes are the same for $\boldsymbol{u}$ and $-\boldsymbol{u}$, but they are different for different eigendirections and the drift direction.

We can determine the transition rates by matching the behaviors of the drift and diffusion part of the SDE. For $\boldsymbol{x}\in \mathbb{S}$, we have
\begin{align*}
	& \sum_{i \in M(\boldsymbol{x})} a_{i, h \boldsymbol{v}_i}^d(\boldsymbol{x}) h \boldsymbol{v}_i(\boldsymbol{x})=\boldsymbol{\mu}(\boldsymbol{x}) \\
	& \sum_{i \in M(\boldsymbol{x})} a_{i, h\boldsymbol{v}_i}^n(\boldsymbol{x}) h^2 \boldsymbol{v}_i(\boldsymbol{x}) \boldsymbol{v}_i^{\top}(\boldsymbol{x})=\boldsymbol{A}(\boldsymbol{x}) \\
	& \sum_{i \in M(\boldsymbol{x})} a_{i, h \boldsymbol{v}_i}^n(\boldsymbol{x}) h \boldsymbol{v}_i(\boldsymbol{x})=0,
\end{align*}
where $\boldsymbol{v}_i(\boldsymbol{x})\in M(\boldsymbol{x})$. By solving the above set of equations, we can obtain the transition rates:
$$\begin{array}{lll}a_{h \boldsymbol{\mu}}^d(\boldsymbol{x})=\frac{1}{h}, & a_{i, \pm h \boldsymbol{u}_i}^d(\boldsymbol{x})=0, & \text { for } i=1, \ldots, d, \\ a_{h \boldsymbol{\mu}}^n(\boldsymbol{x})=0, & a_{i,-h \boldsymbol{u}_i}^n(\boldsymbol{x})=a_{i, h \boldsymbol{u}_i}^n(\boldsymbol{x})=\frac{\lambda_i}{2 h^2}, & \text { for } i=1, \ldots, d.\end{array}$$
Similarly, for $\boldsymbol{x} \in \partial \mathbb{S}$, we replace $\boldsymbol{\mu}(\boldsymbol{x})$ with $\hat{\boldsymbol{\beta}}(\boldsymbol{x})$ and replace $\boldsymbol{A}(\boldsymbol{x})$ with $\hat{\boldsymbol{G}}(\boldsymbol{x})$ in the set of moment matching equations and obtain the corresponding transition rates.

After we construct the CTMC, we can generate paths of stock price processes that exhibit persistent extremes. If $\boldsymbol{x}\in \mathbb{S}$, we have the following steps:

\begin{itemize}
    \item STEP 1: Calculate $a_0$: $a_0 =-a_{\delta_{\boldsymbol{\mu}} \boldsymbol{\mu}}-\sum_{i=0}^{d-1}\left(a_{i,-\delta_i \boldsymbol{u}_{i+1}}+a_{i, \delta_i \boldsymbol{u}_{i+1}}\right)$;
    \item STEP 2: Generate a random variable $e$ from the exponential distribution with mean $\frac{1}{|a_0|}$, and set $t=t+\min\left\{e, T-t\right\}$;
    \item STEP 3: Generate an independent random variable $U$ from the uniform distribution over $[0,1]$, and sample the index $i$:
$$i = \min \left\{j=0, \ldots, 2d: U<\sum_{k=0}^j p_k\right\},$$
where $p_k\in \left[\frac{a_{\delta_{\boldsymbol{\mu}} \boldsymbol{\mu}}}{\left|a_0\right|}, \frac{a_{i,-\delta_i \boldsymbol{u}_{i+1}}}{\left|a_0\right|}, \frac{a_{i, \delta_i\boldsymbol{u}_{i+1}}}{\left|a_0\right|}\right.$ for $i=$
                $\left. 0, \ldots, d-1 \right]; $
    \item STEP 4: Generate the next state $\boldsymbol{y}$ until the maturity $T$:
\[
\boldsymbol{y}=\boldsymbol{x}+\boldsymbol{u}[:,i],
\]
where $\mathbf{u}=\left[\delta_{\boldsymbol{\mu}} \boldsymbol{\mu},-\delta_j \mathbf{u}_{j+1}, \delta_j \mathbf{u}_{j+1}\right.$ for $\left.j=0, \ldots, d-1\right]$.
\end{itemize}

If $\boldsymbol{x}\in \partial\mathbb{S}$, we repeat the above steps with transition rates obtained from $\hat{\boldsymbol{\beta}}(\boldsymbol{x})$ and $\hat{\boldsymbol{G}}(\boldsymbol{x})$.

Figure \ref{fig:dd} and \ref{fig:du} display the sample paths of the CTMC approximating the drawdown process, drawup process and corresponding stock price processes generated by our simulation method. From Figure \ref{fig:dd}, we can see that during the middle time period, the drawdown process exhibits stickiness at one and the stock price continues to rise during that period. From Figure \ref{fig:du}, the drawup process exhibits stickiness at one in different time periods and we can observe record lows for a long stretch of time in the corresponding stock price process. 
\begin{figure}[h]
    \centering
    \includegraphics[width=0.49\textwidth]{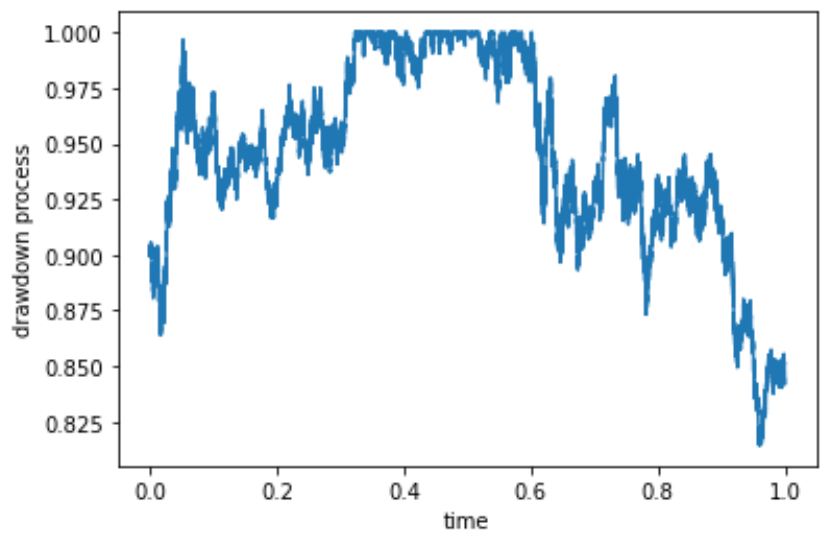}
    \includegraphics[width=0.49\textwidth]{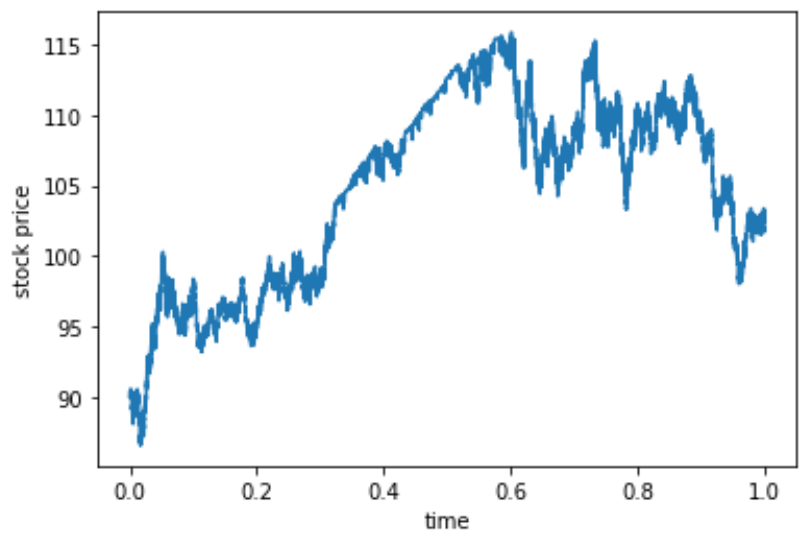}
    \caption{The left panel shows the simulation of the drawdown process $D_t^{\pm}$, and the right panel displays the corresponding stock price process.}
    \label{fig:dd}
\end{figure}
\begin{figure}[h]
    \centering
    \includegraphics[width=0.49\textwidth]{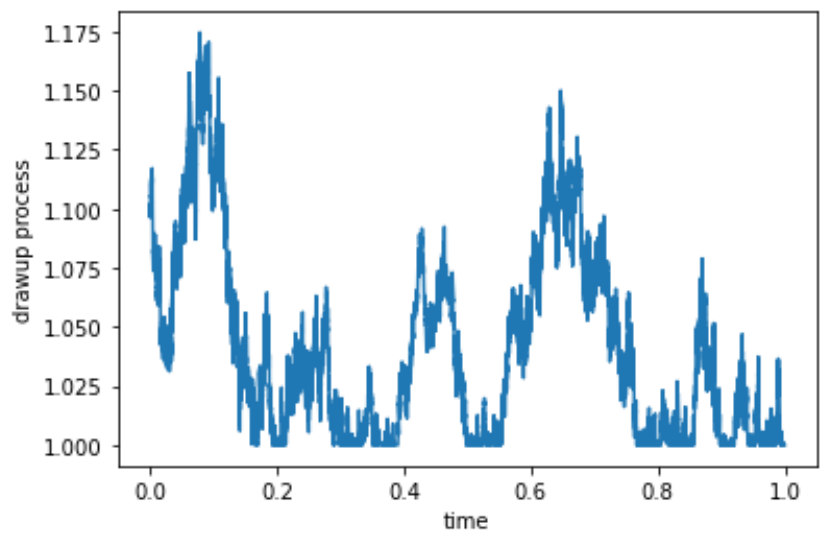}
    \includegraphics[width=0.49\textwidth]{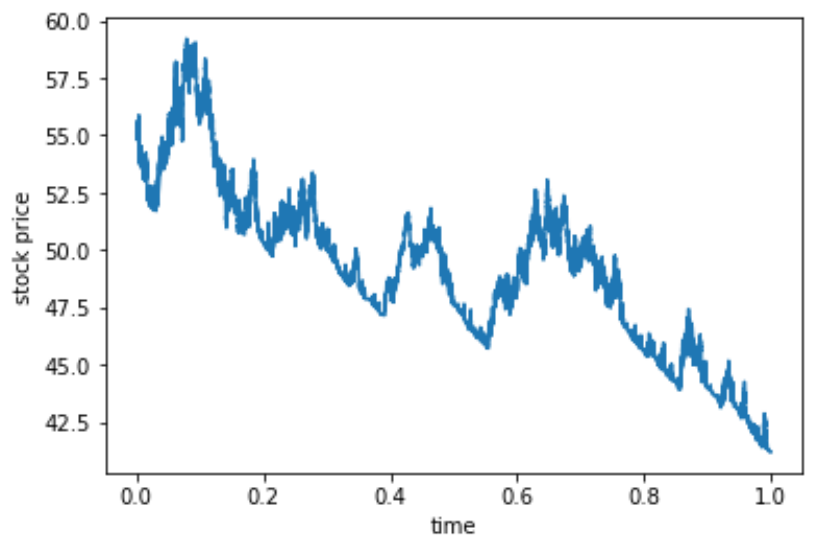}
    \caption{The left panel shows the simulation of the drawup process $U_t^{\pm}$, and the right panel displays the corresponding stock price process.}
    \label{fig:du}
\end{figure}

\subsection{The Pricing PDE}\label{sec:pricing-PDE}
Option prices under the SVSDU model can also be computed by solving the corresponding pricing PDE. For the combined sticky extreme asset process, the price of the European call option under the risk neutral measure is
\begin{align}
	P(t, x, y, z, v;\Phi)=\mathbb{E}^{\mathbb{Q}}\left[e^{-r(T-t)}\left(S_T^{\pm}-K\right)^{+} \mid S_t^{\pm}=x, \overline{S}_t^{\pm}=y, \underline{S}_t^{\pm}=z, V_t=v\right],
\end{align}
where $t$ is the current time, $r$ is the risk-free rate, $x$ is the asset value at time $t$, $y$ is the running maximum of the asset value, $z$ is the running minimum of the asset value, $v$ is the volatility, $T$ is the maturity date and $K$ is the strike price. $\Phi$ is the vector of model parameters. Specifically, $\Phi=(K,r,\rho,\kappa, \theta,\sigma,\eta,T,\xi)\in \mathcal{P}$, where $\mathcal{P}$ is the compact parameter domain. 
So for $0\le t < T$, $0 < z < x < y $, it is shown in Appendix \ref{app-proof} that by Feynman-Kac Theorem, we have the pricing PDE: 
\begin{align}
	\begin{array}{l}
	\frac{1}{2} v x^2 P_{x x}+\rho \sigma v x P_{x v}+\frac{1}{2} \sigma^2 v P_{v v}+r x P_x \\
	+\kappa(\theta-v) P_v+r y P_y+r z P_z+P_t=r P, \quad (t,x,y,z,v)\in \mathcal{Q}.
	\end{array}
\end{align}
Define $\mathcal{L}P=\frac{1}{2} v x^2 P_{x x}+\rho \sigma v x P_{x v}+\frac{1}{2} \sigma^2 v P_{v v}+r x P_x+\kappa(\theta-v) P_v+r y P_y+r z P_z$, then the pricing PDE can be rewritten as:
\[
P_t+\mathcal{L}P-rP=0.
\]
The boundaries are at $x = y$ and $x=z$, and the boundary conditions are,
\begin{align}
	\begin{array}{l}
	P_y(t,y,y,z,v;\Phi)=\left[\frac{1}{2} v y P_{x x}(t,y,y,z,v;\Phi)+\rho \sigma v P_{x v}(t,y,y,z,v;\Phi)\right] \xi, \quad (t,x,y,z,v)\in \Sigma_1 \\
	P_z(t,z,y,z,v;\Phi)=-\left[\frac{1}{2} v z P_{x x}(t,z,y,z,v;\Phi)+\rho \sigma v P_{x v}(t,z,y,z,v;\Phi)\right] \eta, \quad (t,x,y,z,v)\in \Sigma_2.
	\end{array}
\end{align}
The terminal condition is
\begin{align}
	P(T, x, y, z, v;\Phi) = (x-K)^+, \quad (t, x,y,z,v)\in \Omega,
\end{align}
where $0< z\le x\le y$.

We will show how to approximate the solution to the above PDE by deep neural networks in the next section.

\section{Pricing and Calibration with Deep Neural Networks}
\label{sec-deep}
\subsection{Loss Function of The Neural Network}

To train a deep neural network for approximating the solution to the pricing PDE under the SVSDU model, we need to choose an appropriate loss function. We define the deep learning solution as $P^{\theta}(t,x,y,z,v;\Phi)$, which approximates $P(t,x,y,z,v;\Phi)$. $\theta\in \mathbb{R}^d$ contains all trainable network parameters. We can construct the loss function as follows:
\begin{align}\label{eq:loss-function}
\mathcal{J}(P^{\theta})=\omega_{in}\mathcal{J}_{in}(P^{\theta})+\omega_{bc}\mathcal{J}_{bc}(P^{\theta})+\omega_{te}\mathcal{J}_{te}(P^{\theta}),
\end{align}
where
\[
\mathcal{J}_{in}(P^{\theta})=\big|\mathcal{P}\times \mathcal{Q}  \big|^{-1}\int_{\mathcal{P}}\int_{\mathcal{Q}}(P^{\theta}_t+\mathcal{L}P^{\theta}-rP^{\theta})^2 \diff{(t,x,y,z,v)} \diff{\Phi},
\]
\begin{align*}
\mathcal{J}_{bc}(P^{\theta})=\big|\mathcal{P}\times \Sigma_1 \big|^{-1}\int_\mathcal{P}\int_{\Sigma_1}\big[P^{\theta}_y-(\frac{1}{2}vyP^{\theta}_{xx}+\rho\sigma vP^{\theta}_{xv})\xi \big]^2\diff{(t,x,y,z,v)} \diff{\Phi}\\
+\big|\mathcal{P}\times \Sigma_2 \big|^{-1}\int_\mathcal{P}\int_{\Sigma_2}\big[P^{\theta}_z+(\frac{1}{2}vzP^{\theta}_{xx}+\rho\sigma vP^{\theta}_{xv})\eta \big]^2\diff{(t,x,y,z,v)} \diff{\Phi},
\end{align*}
and
\[
\mathcal{J}_{te}(P)=\big|\mathcal{P}\times \Omega \big|^{-1}\int_{\mathcal{P}}\int_{\Omega}\big[P^{\theta}(T,x,y,z,v; \Phi)-(x-K)^+ \big]^2 \diff{(x,y,z,v)}\diff{\Phi}.
\]

If $\mathcal{J}(P^{\theta})=0$, then $P^{\theta}$ is the solution of the PDE. So our goal is to find a set of network parameters $\theta$ such that $\mathcal{J}(P^{\theta})$ is sufficiently small and $P^{\theta}\approx P$. Because the integrals in \eqref{eq:loss-function} are hard to evaluate directly, we numerically evaluate those integrals by Monte–Carlo quadrature with sample points uniformly generated from the domains:
\begin{align*}
\mathcal{J}_{in}(P^{\theta})\approx &\sum_{i=1}^{N} (P^{\theta}_t(t^{(i)}, x^{(i)}, y^{(i)}, z^{(i)}, v^{(i)}; \Phi^{(i)})+\mathcal{L}P^{\theta}(t^{(i)}, x^{(i)}, y^{(i)}, z^{(i)}, v^{(i)}; \Phi^{(i)})\\
&-r^{(i)} P^{\theta}(t^{(i)}, x^{(i)}, y^{(i)}, z^{(i)}, v^{(i)}; \Phi^{(i)})^2/N,   
\end{align*}
\begin{align*}
\mathcal{J}_{bc}(P^{\theta})\approx &\sum_{i=1}^{N}\big[P^{\theta}_y(t^{(i)}, x^{(i)}, y^{(i)}, z^{(i)}, v^{(i)}; \Phi^{(i)})-(\frac{1}{2}v^{(i)} y^{(i)} P^{\theta}_{xx}(t^{(i)}, x^{(i)}, y^{(i)}, z^{(i)}, v^{(i)}; \Phi^{(i)})\\
&+\rho^{(i)} \sigma^{(i)} v^{(i)} P^{\theta}_{xv}(t^{(i)}, x^{(i)}, y^{(i)}, z^{(i)}, v^{(i)}; \Phi^{(i)}))\xi^{(i)} \big]^2/N\\
&+\sum_{i=1}^{N}\big[P^{\theta}_z(t^{(i)}, x^{(i)}, y^{(i)}, z^{(i)}, v^{(i)}; \Phi^{(i)})+(\frac{1}{2}v^{(i)}z^{(i)}P^{\theta}_{xx}(t^{(i)}, x^{(i)}, y^{(i)}, z^{(i)}, v^{(i)}; \Phi^{(i)})\\
&+\rho^{(i)}\sigma^{(i)} v^{(i)} P^{\theta}_{xv}(t^{(i)}, x^{(i)}, y^{(i)}, z^{(i)}, v^{(i)}; \Phi^{(i)}))\eta^{(i)} \big]^2/N,
\end{align*}
\begin{align*}
    \mathcal{J}_{te}(P)\approx \sum_{i=1}^N\big[P^{\theta}(T,x^{(i)},y^{(i)},z^{(i)},v^{(i)}; \Phi^{(i)})-(x^{(i)}-K^{(i)})^+ \big]^2/N.
\end{align*}

\subsection{Structure of The Neural Network}
Our neural network structure is based on the framework proposed by \cite{sirignano2018dgm}, which has been proven effective in solving a range of quasilinear parabolic partial differential equations.

As for the structure of our neural network, the first layer is the input layer:
	\begin{align}
	S^{1} & =\sigma_1\left(\vec{x}W^{1} +b^{1}\right), 
	\end{align}
	where $\vec{x}=(t,x,y,z,v,K,r,\rho, \kappa,\theta,\sigma,\eta,T, \xi)$, $W^{1}\in \mathbb{R}^{d\times m}$, $b^{1}\in \mathbb{R}^{m}$. $\sigma_1(\cdot)$ is the activation function and $\sigma_1(x)=tanh(x)=\frac{e^x-e^{-x}}{e^x+e^{-x}}$.
	
	For $\ell=1, \ldots, L$, where $L$ is the number of layers, we have the hidden layer
	\begin{align}
	S^{\ell+1} & =\left(1-G^{\ell}\right) \odot H^{\ell}+Z^{\ell} \odot S^{\ell}, \quad \ell=1, \ldots, L,
	\end{align}
	where $\odot$ denotes element-wise multiplication. $G^{\ell}, H^{\ell}$ and $Z^{\ell}$ are "sub-layers" of each hidden layer of the neural network. The structure of $G^{\ell}$ and $Z^{\ell}$ are the same:
	\begin{align}
	Z^{\ell} & =\sigma_2\left(\vec{x}U^{z, \ell} +S^{\ell}W^{z, \ell} +b^{z, \ell}\right), \quad \ell=1, \ldots, L, \\
	G^{\ell} & =\sigma_2\left(\vec{x} U^{g, \ell} +S^{\ell} W^{g, \ell} +b^{g, \ell}\right), \quad \ell=1, \ldots, L, \\
	\end{align}
	where $U^{z, \ell}\in \mathbb{R}^{d\times m}$, $U^{g, \ell}\in \mathbb{R}^{d\times m}$, $W^{z, \ell}\in \mathbb{R}^{m\times m}$, $W^{g, \ell}\in \mathbb{R}^{m\times m}$, $b^{z, \ell} \in \mathbb{R}^m$, and $b^{g, \ell}\in \mathbb{R}^m$. As for $H^{\ell}$, the structure is:
	\begin{align}
	H^{\ell} & =\sigma_2\left(\vec{x} U^{h, \ell} +\left(S^{\ell} \odot R^{\ell}\right) W^{h, \ell}+b^{h, \ell}\right), \quad \ell=1, \ldots, L, \\
	\end{align}
	where
	\begin{align}
	R^{\ell} & =\sigma_2\left(\vec{x} U^{r, \ell} +S^{\ell} W^{r, \ell} +b^{r, \ell}\right), \quad \ell=1, \ldots, L. \\
	\end{align}
	In this case $U^{h, \ell} \in \mathbb{R}^{d\times m}$, $U^{r, \ell} \in \mathbb{R}^{d\times m}$, $W^{h, \ell} \in \mathbb{R}^{m\times m}$, $W^{r, \ell}\in \mathbb{R}^{m\times m}$, $b^{h, \ell} \in \mathbb{R}^m$ and $b^{r, \ell}\in \mathbb{R}^m$. The activation function $\sigma_2(x)=tanh(x)=\frac{e^x-e^{-x}}{e^x+e^{-x}}$.
	
	As for the last layer, which generates approximation for the PDE solution, the structure is:
	\begin{align}
	P^{\theta}(\vec{x}) & =S^{L+1} W +b,
	\end{align}
	where $W\in \mathbb{R}^{m\times 1}$, $b \in \mathbb{R}$. 

There is one input layer, four hidden layers ($L=4$) and one output layer in our neural network. Besides, there are 110 nodes on each hidden layer.  The activation function $\sigma_1(x), \sigma_2(x)\in \mathcal{C}^n$ for any $n\in \mathbb{N}$. So according to Theorem 2 in \cite{horvath2021deep}, we can guarantee that the neural network can approximate the first and the second order derivatives of the target function.

\subsection{Network Training}
\label{subsubse-train}
The training process includes two steps: sample training points from the domain and optimize the network parameters by minimizing the loss function.  

There are 14 input variables in the neural network. We set $\omega_{in}=\omega_{te}=\frac{1}{4}$ and $\omega_{bc}=\frac{1}{2}$ in \eqref{eq:loss-function}. 400000 training samples are uniformly sampled from each domain: $\mathcal{P}\times \mathcal{Q}$, $\mathcal{P}\times \Sigma_1$, $\mathcal{P}\times \Sigma_2$ and $\mathcal{P}\times \Omega$. So the total size of the training set is 1600000. We sample $\Phi=(K, r, \rho, \kappa, \theta, \sigma, \eta, T, \xi)$ uniformly on $\mathcal{P}$: $(K, r, \rho, \kappa, \theta, \sigma, \eta, T, \xi)\in \mathcal{U}[50, 131)\times \mathcal{U}[0.01, 0.3)\times \mathcal{U}[-1,1)\times \mathcal{U}[0.01, 5)\times \mathcal{U}[0.01,1)\times \mathcal{U}[0.01,\sqrt{10})\times \mathcal{U}[0.01,10)\times \mathcal{U}[\frac{7}{365}, 1.1)\times \mathcal{U}[0.01,10)$. The Feller condition is satisfied such that $2\kappa\theta>\sigma^2$ when we sample the variables $\kappa$, $\theta$ and $\sigma$. In the domain $\mathcal{Q}$, the variables $t$, $x$, $y$, $z$, $v$ can be sampled in the following way to ensure that $0\le z<x<y$:

\begin{itemize}
    \item STEP 1: Sample a new variable $m$ uniformly from the region $[50, 131)$;
    \item STEP 2: Sample $z$ and $y$ uniformly from the region $[1,m-1)$ and $[m+1,m+100)$;
    \item STEP 3: Sample $x$ uniformly from the region $[z+1, y)$;
    \item STEP 4: Sample $t$ and $v$ uniformly from the region $[0, T)$ and $[0.01, 0.16)$.
\end{itemize}

Similarly, when we sample $(t,x,y,z,v)$ in the domain $\Sigma_1$ and $\Sigma_2$, we set $x=y$ and $x=z$ in the step 3 and the remaining steps are the same. As for the domain $\Omega$, we set $t=T$ in the step 4 and the remaining steps are the same. 

The ranges of the training sample variables $x$, $y$, $z$ and $K$ should not be too large, since a large range of those input variables will lead to a large training loss of the neural network, causing erratic updates of the network weights and poor convergence of the neural network. Besides, in order to avoid vanishing gradients and accelerate training process, we transform the domains of all the 14 input variables linearly to $[-1,1]$.

As for the optimization step, we use ADAM algorithm in \cite{kingma2014adam} to update parameters of the neural network. In order to help the neural network optimization algorithms converge more effectively, we decrease the learning rate during training and the learning rate is a piecewise constant function of the number of iterations: 
$$\alpha_n= \begin{cases}10^{-3} & n \leq 1,000 \\ 5 \times 10^{-4} & 1,000<n \leq 2,000 \\ 10^{-4} & 2,000<n \leq 3,000 \\ 5 \times 10^{-5} & 3,000<n \leq 4,000 \\ 10^{-5} & 4,000<n \leq 5,000 \\ 5 \times 10^{-6} & 5,000<n \leq 5,500 \\ 10^{-6} & 5,500<n.\end{cases}$$
We use 6000 iterations and the batch size is 10000.
 
\subsection{The Calibration Steps}
\label{subsec-cali-step}

After we train the neural network, we can calibrate parameters of the model approximated by the deep neural network. In the calibration process, we need to solve the following optimization problem on day $d$:
$$\hat{\phi}:=\underset{\phi \in \Theta}{\operatorname{argmin}} \sum_{i=1}^{N_d}\left(P^{\theta}(t, x_d, K^d_i, T^d_i; \phi)-P^{M K T}\left(t, x_d, K^d_i, T^d_i\right)\right)^2,\quad d=1, 2, \cdots, n,$$
where $x_d$ represents underlying asset price on day $d$, $K^d_i$ and $T^d_i$ represent strike price and maturity date on day $d$ for the $i$-th option, $P^{\theta}$ is the neural network approximation for the option price, $P^{M K T}$ is the market price of the option, and $\phi=(\rho, \kappa, \theta, \sigma, \eta, v, \xi, y, z)$ is the set of parameters that we need to calibrate. 

When we train the neural network to approximate the SVSDU model, input variables such as asset prices and strike prices are restricted to a bounded domain. However, in the financial market, the magnitudes of the underlying asset prices and strike prices often significantly exceed this domain, adversely affecting the calibration results. If we perform calibration tasks with the market data, it is essential to ensure that the magnitudes of the input data values align with the magnitudes of the training samples used to train the neural network. So we standardize the market data by dividing the underlying asset prices, the strike prices and the corresponding option prices by the same real number during calibration. The optimization problem thus becomes:
\begin{align}\label{eq:modified-optimization-problem}
\hat{\phi}:=\underset{\phi \in \Theta}{\operatorname{argmin}} \sum_{i=1}^{N_d}\left(P^{\theta}(t, \frac{x_d}{l_d}, \frac{K_i}{l_d}, T_i; \phi)-\frac{P^{M K T}\left(t, x_d, K_i, T_i\right)}{l_d}\right)^2,\quad d=1, 2, \cdots, n,
\end{align}
where $\phi=(\rho, \kappa, \theta, \sigma, \eta, v, \xi, \frac{y}{l_d}, \frac{z}{l_d})$, and $l_d$ is the scaling factor on day $d$. In our experiments, we set $l_d=\frac{x_d}{C_{l_d}}$, where $C_{l_d}$ is a predetermined value within the range of $(0, 232)$. We use Levenberg-Marquadt algorithm (see \cite{gavin2019levenberg}) to solve the above minimization problem and obtain parameter set $\hat{\phi}$. The details of the Levenberg-Marquadt algorithm can be seen in Appendix \ref{app-cali}.

Since we will obtain different $\hat{\phi}$ with different $l_d$, we need to find the optimal value of $l_d$ during calibration to find the optimal $\hat{\phi}$. As $l_d$ decreases or increases, $\frac{x_d}{l_d}$ and $\frac{K_i}{l_d}$ will approach the boundaries of the range of input variables of the neural network, resulting in a worse performance of the neural network in calibration. So in order to find the optimal $l_d$, we can find the optimal $C_{l_d}$ by following steps:

\begin{itemize}
    \item STEP 1: Set the initial value of $C_{l_d}$;
    \item STEP 2: Increase or decrease the value of $C_{l_d}$ with a step size $h$ and use the Levenberg-Marquadt algorithm to calibrate the neural network model until we find an $C_{l_d}$ that the calibration error with the scaling factor $l_d=\frac{x_d}{C_{l_d}}$ is smaller than the calibration errors with scaling factors $l_d=\frac{x_d}{C_{l_d}-h}$ and $l_d=\frac{x_d}{C_{l_d}+h}$ respectively or the number of searches reaches the predefined maximum number of times;
    \item STEP 3: Record the optimal value of $C_{l_{d}}$. 
\end{itemize}

In our experiment, we set $h=10$. Since we perform the calibration tasks on a daily basis, we set the initial value of $C_{l_d}$ as the optimal value of $C_{l_{d-1}}$ on day $d-1$ to accelerate the calibration algorithm.

\section{Numerical Results on Simulated Data}
\label{sec-num}
In this section, we test the performances of our models with simulated data generated by the simulation method in section \ref{sec-simu}. We perform the numerical experiments on a workstation with two CPUs of Intel Xeon Platinum 8171 2.6 GHz and a GPU of NVDIA RTX6000. 

\subsection{Pricing Accuracy and Speed}
In this part, we demonstrate the pricing accuracy and speed of our neural network as the approximation of the true pricing functionals under the SVSDU model. We compute the option prices with different input parameters by Monte Carlo method and the deep neural network. All of the input values are within the ranges of the training samples of the neural networks that are shown in section \ref{subsubse-train}. As for the Monte Carlo method, we generate 60000 paths of CTMC approximating the asset value process to compute the option price for one parameter set, and we set the step size $h=\frac{1}{100}$. We compute Monte Carlo prices across 1219 random parameter combinations and use them as benchmarks. The average percentage error of the  neural network prices is 0.3596$\%$, and the average absolute error is 0.02734, which are small. The average percentage error (APE) and the average absolute error (AAE) are calculated as follows:
\[
\text{APE}=\frac{\sum^N_{i=1}\big|P^{\theta}(t,x,K_i,T_i;\phi)-P^{MKT}(t,x,K_i,T_i)  \big|}{\sum^N_{i=1}P^{MKT}(t,x,K_i,T_i)},
\]
\[
\text{AAE}=\frac{\sum^N_{i=1}\big|P^{\theta}(t,x,K_i,T_i;\phi)-P^{MKT}(t,x,K_i,T_i)  \big|}{N}.
\]
where $P^{MKT}(t,x,K_i,T_i)$ is the market price of the $i$-th option, and the $P^{\theta}(t,x,K_i,T_i;\phi)$ is the corresponding option price generated by our neural network. 
Figure \ref{fig-com} shows that the option prices given by the deep neural network match the values from the Monte Carlo method with different parameter sets. The accuracy of the neural network in pricing guarantees its accuracy in the model calibration.      
\begin{figure}[!h]
		\small
		\centering
		\subfigure[$T=0.3, t=0, y=101, z=49,r=0.04, K=70, v=0.05, \rho=-0.3, \sigma=0.4, \kappa=3, \theta=0.05, \xi=3, \eta=0.7$]{
			\begin{minipage}[t]{0.5\linewidth}
				\centering
				\includegraphics[width=3.2in]{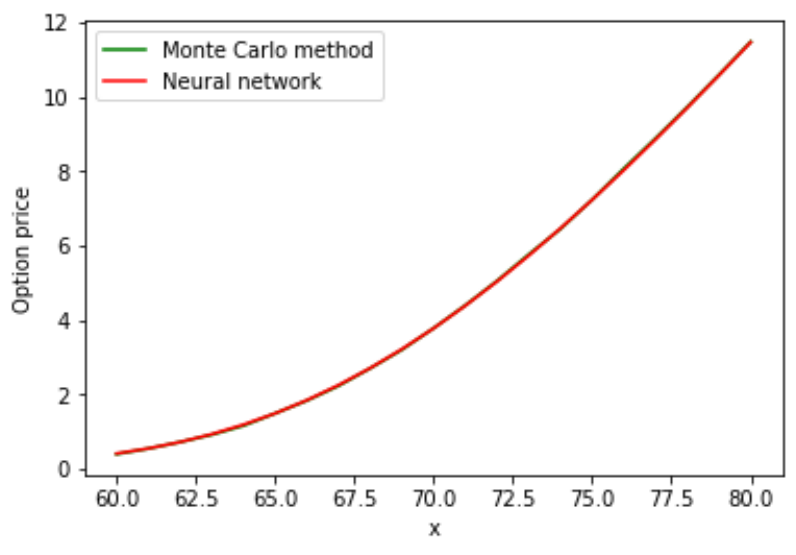}
			\end{minipage}%
            \hspace{0.25cm}
		}%
		\subfigure[$T=0.3, t=0, y=120, z=60,r=0.04, K=85, v=0.03, \rho=-0.3, \sigma=0.4, \kappa=3, \theta=0.06, \xi=2, \eta=0.9$]{
			\begin{minipage}[t]{0.5\linewidth}
				\centering
				\includegraphics[width=3.2in]{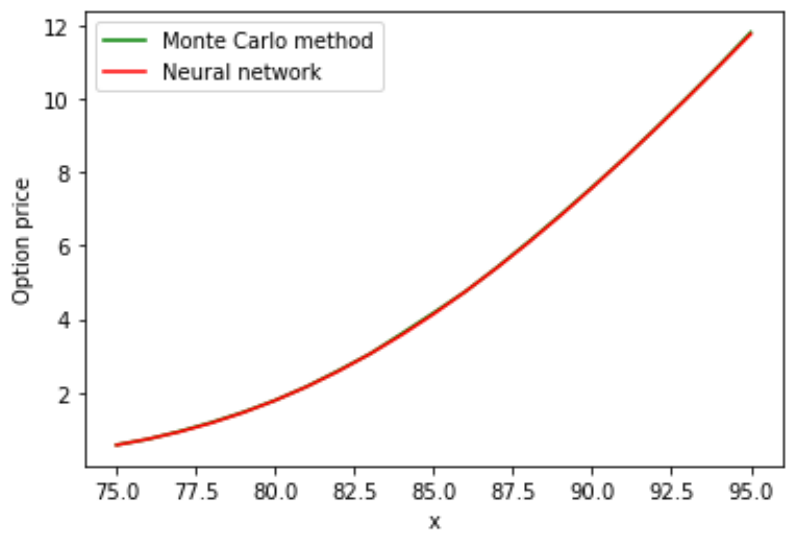}
			\end{minipage}%
		}%
  
  \subfigure[$T=0.36, t=0, y=120, z=60,r=0.06, K=90, v=0.05, \rho=-0.7, \sigma=0.6, \kappa=3, \theta=0.07, \xi=4, \eta=2$]{
			\begin{minipage}[t]{0.5\linewidth}
				\centering
				\includegraphics[width=3.2in]{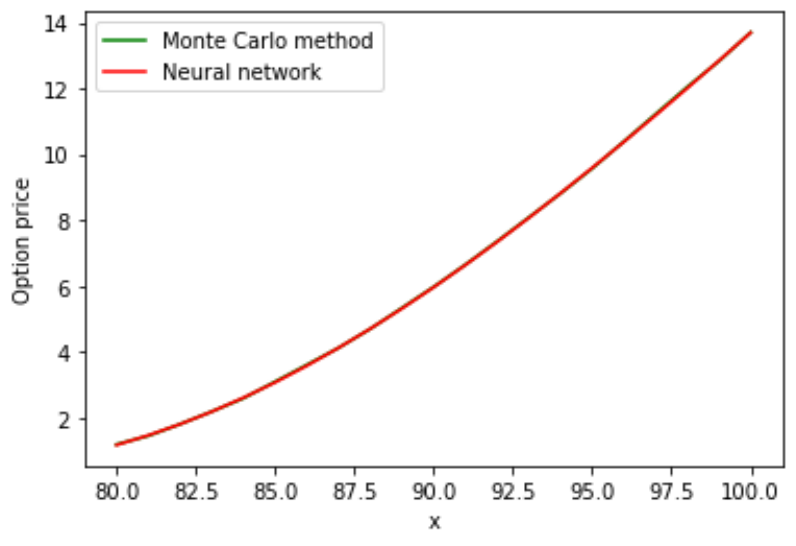}
			\end{minipage}%
   \hspace{0.25cm}
		}%
  \subfigure[$T=0.45, t=0, y=120, z=60,r=0.07, K=85, v=0.03, \rho=-0.3, \sigma=0.4, \kappa=3, \theta=0.05, \xi=3, \eta=0.9$]{
			\begin{minipage}[t]{0.5\linewidth}
				\centering
				\includegraphics[width=3.2in]{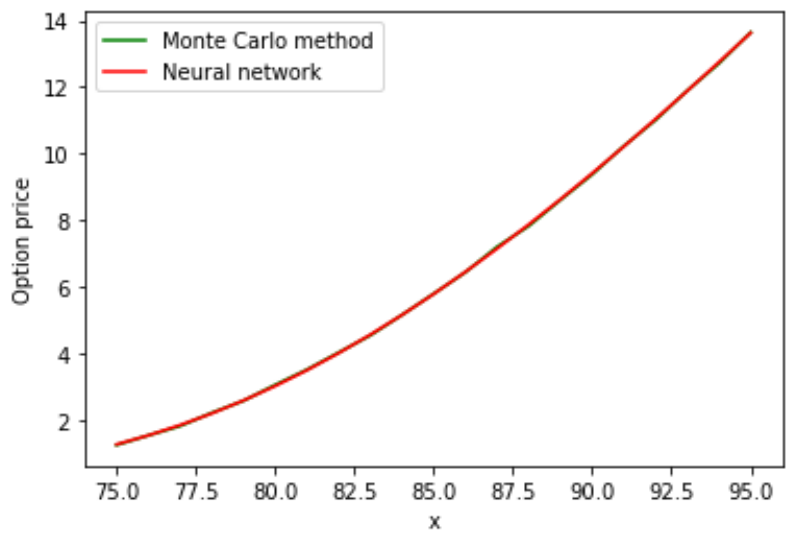}
			\end{minipage}%
		}%

  \subfigure[$T=0.24, t=0, y=100, z=50,r=0.025, K=75, v=0.05, \rho=-0.6, \sigma=0.1, \kappa=1.2, \theta=0.04, \xi=2.5, \eta=4$]{
			\begin{minipage}[t]{0.5\linewidth}
				\centering
				\includegraphics[width=3.2in]{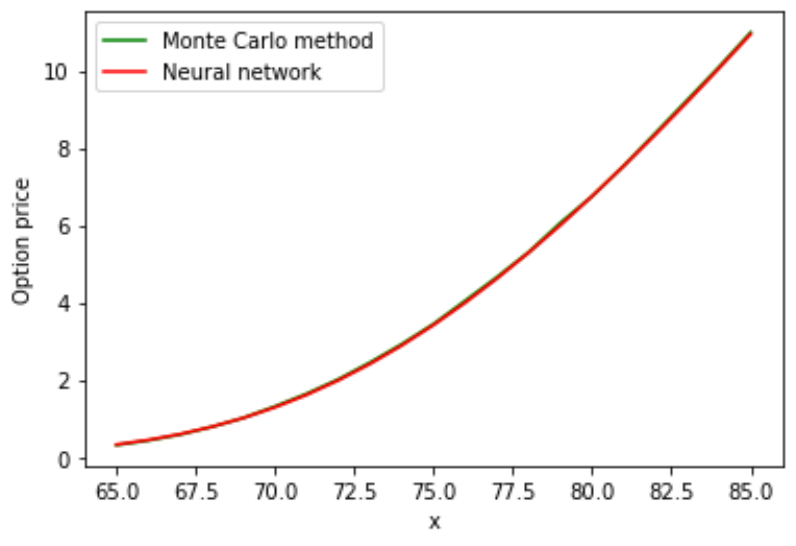}
			\end{minipage}%
   \hspace{0.25cm}
		}%
  \subfigure[$T=0.32, t=0, y=100, z=40, r=0.03, K=70, v=0.06, \rho=-0.6, \sigma=0.5, \kappa=3, \theta=0.08, \xi=5, \eta=2$]{
			\begin{minipage}[t]{0.5\linewidth}
				\centering
				\includegraphics[width=3.2in]{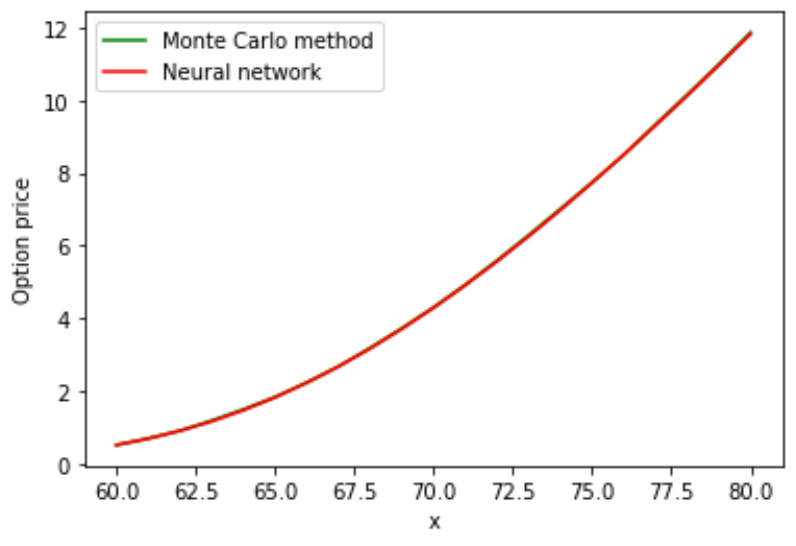}
			\end{minipage}%
		}%
  \centering
  \caption{Comparison of European option prices under the SVSDU model given by the neural network and those given by the Monte Carlo method.}
		\label{fig-com}
	\end{figure}

Table \ref{tab-time} shows the computational time required for pricing a European option using Monte Carlo method and neural networks under the SVSDU model. We consider time to maturity $\tau=0.3$, $\tau=0.2$ and $\tau=0.1$. The numbers of paths used in the Monte Carlo method are 1,000, 6,000, 10,000, and 60,000, respectively. We can see that compared with the Monte Carlo method, the neural network has a significant advantage in terms of speed in pricing options under the SVSDU model.

\begin{table}[h]
\resizebox{\linewidth}{!}{
\begin{tabular}{lccccc}
\hline
& \begin{tabular}[c]{@{}c@{}}Monte Carlo \\ method (1000 paths)\end{tabular} & \begin{tabular}[c]{@{}c@{}}Monte Carlo \\ method (6000 paths)\end{tabular} & \begin{tabular}[c]{@{}c@{}}Monte Carlo \\ method (10000 paths)\end{tabular} & \begin{tabular}[c]{@{}c@{}}Monte Carlo \\ method (60000 paths)\end{tabular} & Neural network \\ \hline
\multicolumn{1}{c}{$\tau$=0.3} & 21s       & 123s   & 206s  & 1241s     & 0.0181s        \\ \hline
$\tau$=0.2   & 15s                                          & 89s     & 144s     & 872s    & 0.0180s        \\ \hline
$\tau$=0.1    & 8s    & 48s    & 78s     & 469s                  & 0.0181s        \\ \hline
\end{tabular}}
\caption{Computational time required for pricing an option via neural networks and the Monte Carlo method}
\label{tab-time}
\end{table}

\subsection{Calibration Accuracy and Speed}

In order to asses the calibration accuracy of our method, we calibrate the SVSDU model to simulated data and compute the average absolute error between the calibrated model parameter $\phi^*$ and the corresponding parameter $\bar{\phi}$ that was chosen for the generation of the synthetic data:
\[
\text{Error}=\frac{|\phi^*-\bar{\phi}|}{n},
\]
where $n$ is the number of test cases.

We sample the true model parameters within the input range of our neural network and generate option prices as test cases with those parameters. So we do not need to find the optimal scaling factor. Starting from the initial guess of model parameters $\phi_0$, we obtain the calibrated model parameters $\phi^*$ and compute the average absolute error.

Table \ref{tab-cali-err} shows the performance of our calibration method. We find that the average absolute errors for all of the calibrated parameters are small, and the calibration result for $v^*$ is the most precise among all parameters. Compared to previous research, the calibration of the SVSDU model with our neural network is more complex. For one thing, there are 9 model parameters that need to be calibrated. For another, calibrating to the option prices will result in a great disparity in sensitivity of different parameters compared to the implied volatility according to \cite{liu2019neural}, making the calibration problem increasingly complex. So the results reported in Table \ref{tab-cali-err} can be considered satisfactory. Besides, the average CPU time for the calibration process is 164.2 milliseconds, which demonstrates the efficiency of our calibration algorithm.

\begin{table}[h]
\resizebox{\linewidth}{!}{
\begin{tabular}{ccccccccc}
\hline
$|\rho^*-\bar{\rho}|$  & $|\kappa^*-\bar{\kappa}|$ & $|\theta^*-\bar{\theta}|$ & $|\sigma^*-\bar{\sigma}|$ & $|\eta^*-\bar{\eta}|$  & $|v^*-\bar{v}|$        & $|\xi^*-\bar{\xi}|$    & $|DD^*-\bar{DD}|$      & $|DU^*-\bar{DU}|$      \\ \hline
$1.56\times 10^{-2}$ & $6.51\times 10^{-2}$    & $8.34\times 10^{-3}$    & $1.99\times 10^{-2}$    & $4.07\times 10^{-2}$ & $1.66\times 10^{-3}$ & $3.20\times 10^{-2}$ & $1.10\times 10^{-2}$ & $7.81\times 10^{-2}$ \\ \hline
\end{tabular}}
\caption{Average absolute errors between the calibrated model parameter and the true model parameters. Since option prices depend on the drawdown ratio $DD$ and the drawup ratio $DU$, we assess the calibration accuracy of $DD$ and $DU$ instead of $y$ and $z$}
\label{tab-cali-err}
\end{table}

\section{Empirical Studies}
\label{sec-emp}

\subsection{Calibration with historical data}
In this section, we perform calibration tasks using neural network approximators. To explore the joint effect of the drawdown and drawup stickiness factors, we propose two other models as comparisons, the Sticky Drawdown Stochastic Volatility (SVSD) model with only drawdown stickiness factor and the Sticky Drawup Stochastic Volatility (SVSU) model with only drawup stickiness factor (see Appendix \ref{app-SVSDSVSU}). Similar with the SVSDU model, we train two neural networks to approximate the pricing functions under the SVSD model and the SVSU model, respectively. We also compare the calibration results with the Heston model. We use Fourier methods to price options under the Heston model, and then calibrate the model to historical data.

We want to explore whether our models perform differently in different market situations, so we consider SPX European call options in 2021 and 2022. Figure \ref{fig-spx-index} shows daily closing values of S$\&$P 500 index in 2021 and 2022\footnote{https://finance.yahoo.com/quote/\%5ESPX/history/}. In 2021, most of time  the S$\&$P 500 index value reported record highs. While in 2022, both winning and losing streaks appeared frequently, and there was a period of significant and sustained decline in the index in the first half of 2022. 

\begin{figure}[!h]
		\small
		\centering
		 \subfigure[S$\&$P 500 index closing values in 2021.]{
			\begin{minipage}[t]{0.5\linewidth}
				\centering
				\includegraphics[width=3.0in]{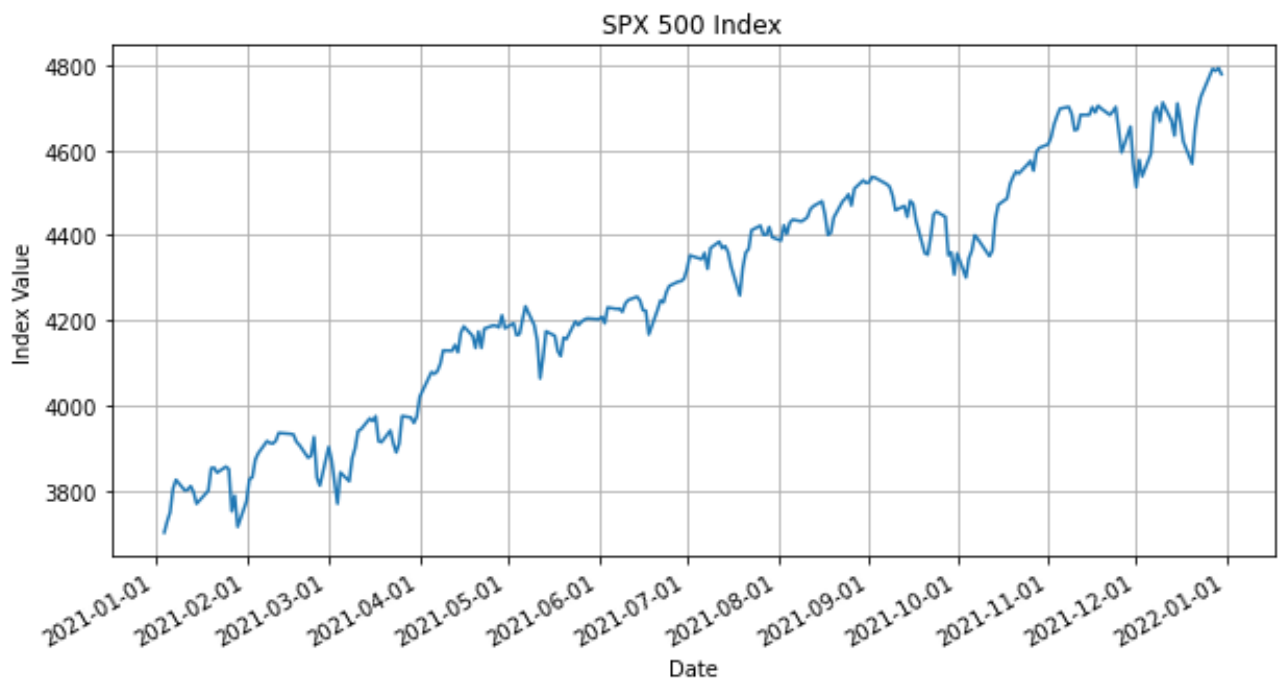}
			\end{minipage}%
            \hspace{0.25cm}
		}%
  \subfigure[S$\&$P 500 index closing values in 2022.]{
			\begin{minipage}[t]{0.5\linewidth}
				\centering
				\includegraphics[width=3.0in]{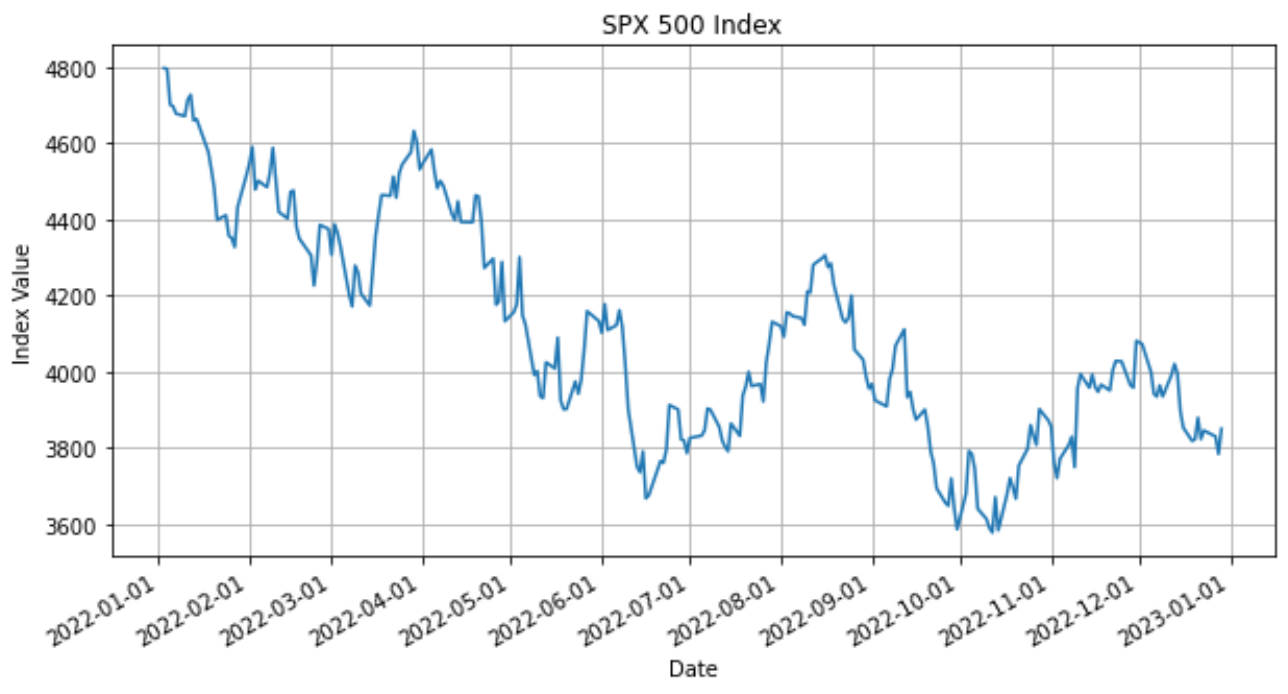}
			\end{minipage}%
   \hspace{0.25cm}
		}%
\caption{S$\&$P 500 index closing values in 2021 and 2022.}
\label{fig-spx-index}
        \end{figure}

We apply some filters to construct the option data for calibration. First, we choose options with maturities more than 6 days and less than one year for calibration. We standardize the time to maturities by dividing each of them by 365. Second, to mitigate the impact of price discreteness on option valuation, options with price quotes lower than \$1 are not considered (see \cite{eraker2004stock}). Finally, we choose quotes that satisfy the arbitrage restriction (see \cite{bakshi1997empirical}):
\[
P^{MKT}(t, x, K, T)\ge \max(0, x-Ke^{-r(T-t)}).
\]

The sample properties of the option data in two different periods are described in Table \ref{tab-sample}, which shows the statistics for the average bid-ask mid-point price, the average  effective bid-ask spread (ask price minus the bid-ask mid-point) which are shown in parentheses, and the total number of observations for each moneyness-maturity category (in braces) in different periods. There are a total of 2970505 call option observations in all of the periods, with ITM and ATM options respectively taking up 51.57\% and 22.42\% of the total sample. 

\begin{table}[!h]
\begin{center}
\begin{tabular}{lcccccc}
\hline
                         \multicolumn{1}{l}{}  & \multicolumn{1}{l}{} & \multicolumn{1}{l}{}                                    & \multicolumn{3}{c}{Days-to-Expiration}                                                         & \multicolumn{1}{l}{}         \\ \cline{4-6}
\multicolumn{1}{c}{Sampling period}  & \multicolumn{1}{l}{} & \begin{tabular}[c]{@{}c@{}}Moneyness\\ S/K\end{tabular} & \multicolumn{1}{c}{\textless{}60} & \multicolumn{1}{c}{60-180} & \multicolumn{1}{c}{$\ge$ 180} & \multicolumn{1}{c}{Subtotal} \\ \hline
04/01/2021-31/12/2021  & OTM                  & \textless{}0.94                                         & \$3.74                            & \$19.02                     & \$41.38                       &                              \\
&                      &                                                         & (0.14)                             & (0.35)                      & (1.64)                         &                         \\
                         &                      &                                                         & \{20554\}                             & \{82533\}                      & \{47690\}                         & \{150777\}                        \\
                         &                      & 0.94-0.97                                               & \$9.64                            & \$59.37                    & \$149.74                       &                              \\
                         &                      &                                                         & (0.19)                             & (0.52)                      & (2.85)                          &                         \\
                         &                      &                                                         & \{58782\}                             & \{53597\}                      & \{11206\}                          & \{123585\}                        \\
                         & ATM                  & 0.97-1.00                                               & \$33.11                           & \$118.75                    & \$221.097                      &                              \\
                         &                      &                                                         & (0.28)                             & (0.63)                      & (3.39)                          &                         \\
                         &                      &                                                         & \{99585\}                           & \{53063\}                       & \{10598\}                           & \{163246\}                         \\
                         &                      & 1.00-1.03                                               & \$106.57                             & \$198.20                   & \$300.93                      &                              \\
                         &                      &                                                         & (0.55)                             & (0.91)                      & (3.83)                          &                         \\
                         &                      &                                                         & \{99323\}                             & \{50144\}                      & \{9920\}                          & \{159387\}                        \\
                         & ITM                  & 1.03-1.06                                               & \$207.72                            & \$287.12                   & \$384.23                      &                              \\
                         &                      &                                                         & (1.12)                             & (1.26)                      & (4.04)                          &                         \\
                         &                      &                                                         & \{89738\}                             & \{47201\}                      & \{9381\}                          & \{146320\}                        \\
                         &                      & $\ge$ 1.06                                              & \$651.53                            & \$741.60                   & \$957.03                      &                              \\
                         &                      &                                                         & (2.64)                             & (2.88)                      & (5.53)                         &                        \\
                         &                      &                                                         & \{307983\}                             & \{337734\}                      & \{97978\}                         &   \{743695\}                     \\
                         & Subtotal             &                                                         & \{675965\}                            & \{624272\}                     &\{186773\}              & \{1487010\}                       \\
                     03/01/2022-30/12/2022                 & OTM                  & \textless{}0.94                                         & \$8.01                            & \$30.09                    & \$64.54                       &                              \\
                     &                      &                                                         & (0.14)                             & (0.34)                      & (1.14)                         &                        \\
                      &                      &                                                         & \{91751\}                             & \{177210\}                      & \{79484\}                         & \{348445\}                       \\
                     &                      & 0.94-0.97                                               & \$25.19                           & \$109.88                   & \$235.78                      &                              \\
                     &                      &                                                         & (0.24)                             & (0.59)                      & (2.11)                          &                         \\
                     &                      &                                                         & \{86146\}                              & \{52702\}                       & \{10872\}                           & \{149720\}                         \\
                     & ATM                  & 0.97-1.00                                               & \$61.15                           & \$172.37                   & \$306.33                      &                              \\
                     &                      &                                                         & (0.34)                             & (0.69)                      & (2.32)                          &                         \\
                     &                      &                                                         & \{110924\}                             & \{51567\}                      & \{10229\}                          & \{172720\}                        \\
                     &                      & 1.00-1.03                                               & \$128.12                          & \$245.06                   & \$379.66                      &                              \\
                     &                      &                                                         & (0.73)                             & (0.89)                      & (2.59)                          &                         \\
                     &                      &                                                         & \{111335\}                            & \{49641\}                      & \{9702\}                          & \{170678\}                       \\
                     & ITM                  & 1.03-1.06                                               & \$214.26                          & \$323.46                   & \$454.63                      &                              \\
                     &                      &                                                         & (2.03)                              & (1.60)                      & (3.00)                          &                         \\
                     &                      &                                                         & \{93805\}                               & \{46738\}                       & \{9172\}                           & \{149715\}                         \\
                     &                      & $\ge$ 1.06                                              & \$523.86                          & \$642.32                   & \$927.37                      &                              \\
                     &                      &                                                         & (3.35)                             & (4.05)                     & (5.83)                         &                        \\
                     &                      &                                                         & \{178523\}                              & \{230121\}                      & \{83573\}                          & \{492217\}                        \\
                     & Subtotal             &                                                         & \{672484\}                               & \{607979 \}                      & \{203032 \}                         &\{1483495 \}                        \\\hline
\end{tabular}
\end{center}
\caption{Sample properties of S\&P 500 index options in three different periods}
\label{tab-sample}
\end{table}

As for the in-sample calibration, on each day $d$ we obtain $P^{MKT}$ by:
\[
P_i^{MKT}=\frac{P_i^{Bid}+P_i^{Ask}}{2}, 
\]
where $P_i^{Bid}$ is the bid price and $P_i^{Ask}$ is the ask price of the $i$-th option. Interest rates are obtained from the daily one year treasury  yield data published by the Federal Reserve Board based on the average yield of a range of Treasury securities, which are all adjusted to the equivalent of a one-year maturity. We estimate parameters $(\rho, \kappa, \theta, \sigma, \eta, v, \xi, y, z)$ for the SDUDV model, $(\rho, \kappa, \theta, \sigma,  v, \xi, y)$ for the SVSD model, $(\rho, \kappa, \theta, \sigma, \eta, v, z)$ for the SVSU model and ($\rho$, $\sigma$, $\theta$, $\kappa$, $v$) for the Heston model. Table \ref{tab-in-error} shows the in-sample calibration errors for the four models in two time periods. The errors are calculated after we recover the values of option prices by multiplying the corresponding scaling factor. From Table \ref{tab-in-error}, we can find that the SVSDU model, the SVSU model and the SVSD model outperform the Heston model in-sample in all the periods, which implies that the stickiness factors can enhance the model's fit.  Besides, the SVSD model performs better than the SVSU model in-sample in 2021. In 2022, although the index experienced frequent fluctuations in both upward and downward directions, the frequency of index declines was relatively higher. So the SVSU model slightly outperforms the SVSD model in-sample in 2022, but the gap is minor. The findings show that the drawdown stickiness coefficient can improve the model's fit in the bull market, while the drawup stickiness coefficient makes model fit the data in the bear market better. Moreover, the fact that the SVSDU model is the best performer in all of the market conditions shows that the combination of drawdown and drawup stickiness coefficient can improve the model's in-sample fit further.

\begin{table}[!h]
\begin{center}
\begin{tabular}{cccc}
\hline
Sampling period       & Model  & APE       & AAE    \\ \hline
04/01/2021-31/12/2021 & Heston & 0.7410\% & 3.1164 \\
                      & SVSD   & 0.6610\% & 2.7786 \\
                      & SVSU   & 0.6960\% & 2.9274 \\
                      & SVSDU  & 0.5380\% & 2.2616 \\
03/01/2022-30/12/2022 & Heston & 1.1140\%  & 3.2213 \\
                      & SVSD   & 0.8060\%  & 2.3295 \\
                      & SVSU   & 0.8050\%  & 2.3271 \\
                      & SVSDU  & 0.6970\%  & 1.9619 \\ \hline
\end{tabular}
\end{center}
\caption{In-sample calibration errors for four models in two periods}
\label{tab-in-error}
\end{table}

To analyse the in-sample calibrations in more details, we check the fits with respect to  Black–Scholes implied volatility for some expiries on different dates in 2021 and 2022. The in-sample calibrated implied volatility curves of the four models and the mid implied volatility curves derived from mid prices are shown in Figure \ref{fig-imp}. It is obvious that the SVSDU model fits the implied volatilities for SPX options pretty well. The SVSD model and the SVSU model fit the SPX option smiles less well than the SVSDU model, but perform better than the Heston model.

\begin{figure}[!h]
		\small
		\centering
		\subfigure[$\tau=\frac{32}{365}$]{
			\begin{minipage}[t]{0.32\linewidth}
				\centering
				\includegraphics[width=2in]{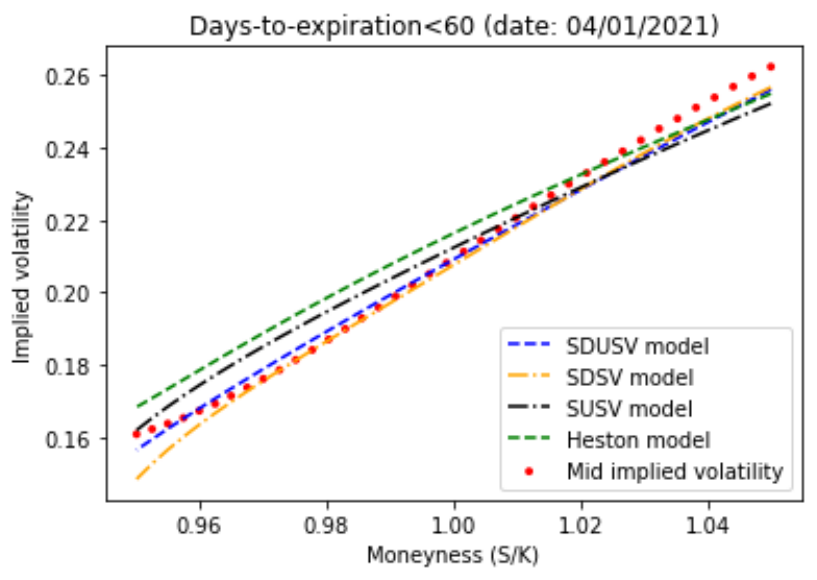}
			\end{minipage}%
		}%
		\subfigure[$\tau=\frac{149}{365}$]{
			\begin{minipage}[t]{0.32\linewidth}
				\centering
				\includegraphics[width=2in]{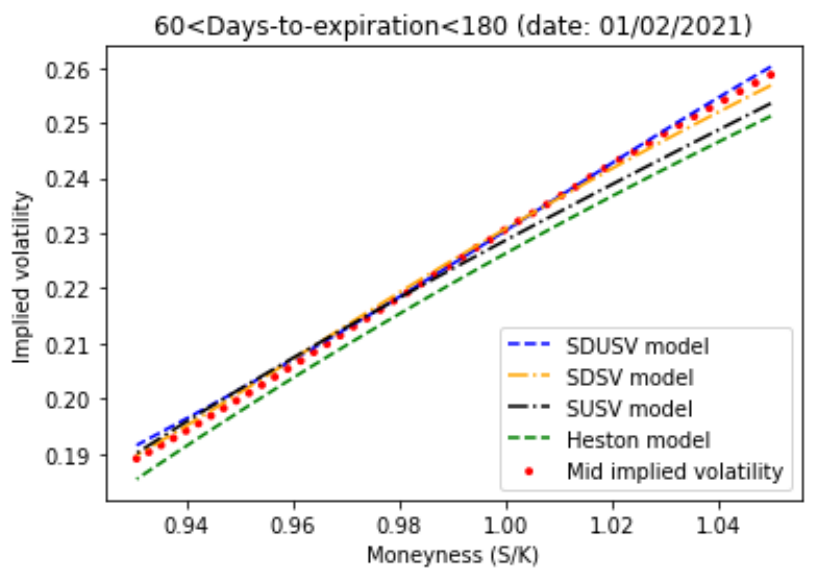}
			\end{minipage}%
		}%
  \subfigure[$\tau=\frac{325}{365}$]{
			\begin{minipage}[t]{0.32\linewidth}
				\centering
				\includegraphics[width=2in]{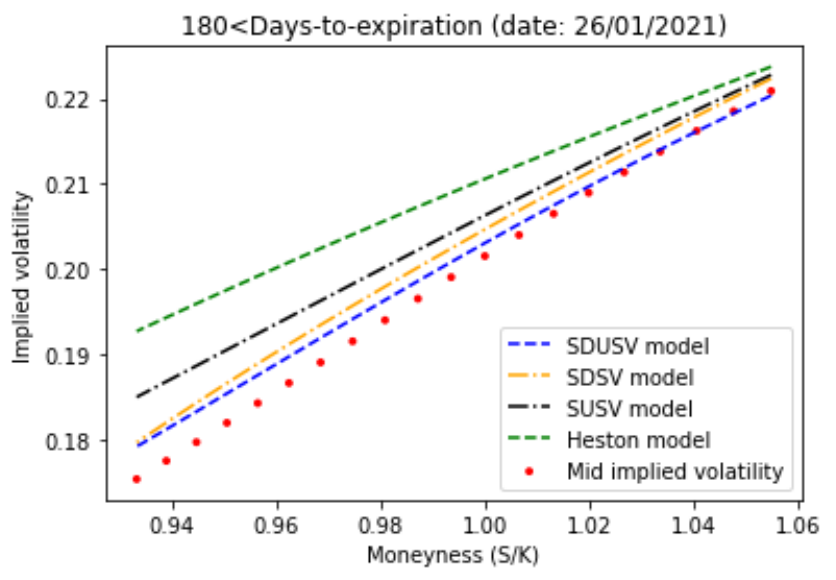}
			\end{minipage}%
		}%

  \subfigure[$\tau=\frac{42}{365}$]{
			\begin{minipage}[t]{0.32\linewidth}
				\centering
				\includegraphics[width=2in]{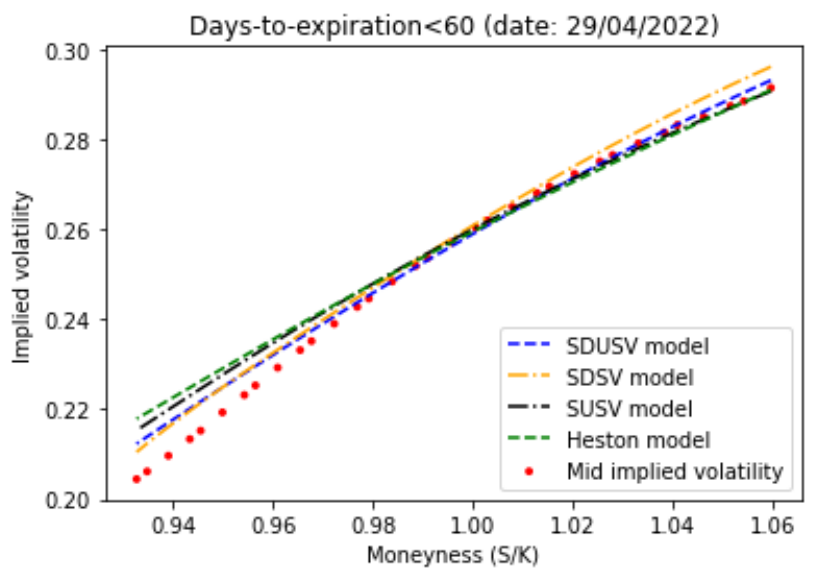}
			\end{minipage}%
		}%
\subfigure[$\tau=\frac{76}{365}$]{
			\begin{minipage}[t]{0.32\linewidth}
				\centering
				\includegraphics[width=2in]{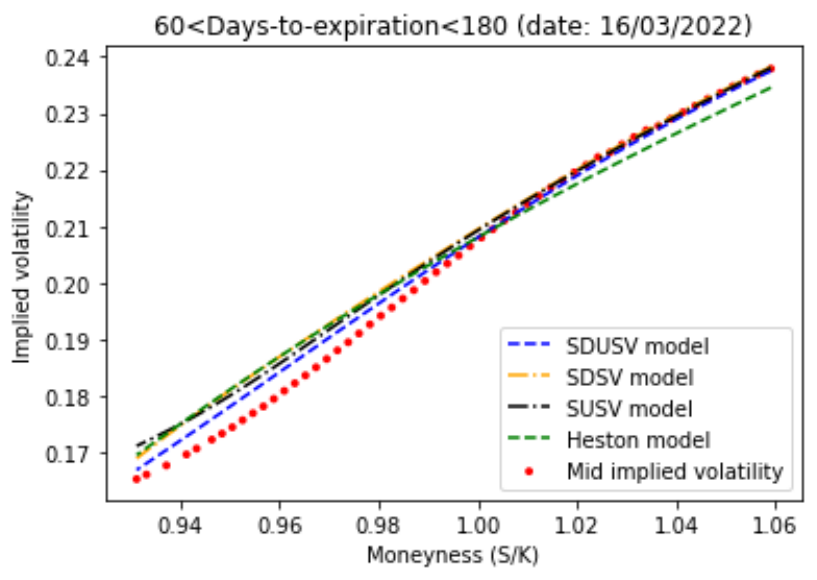}
			\end{minipage}%
		}%
  \subfigure[$\tau=\frac{186}{365}$]{
			\begin{minipage}[t]{0.32\linewidth}
				\centering
				\includegraphics[width=2in]{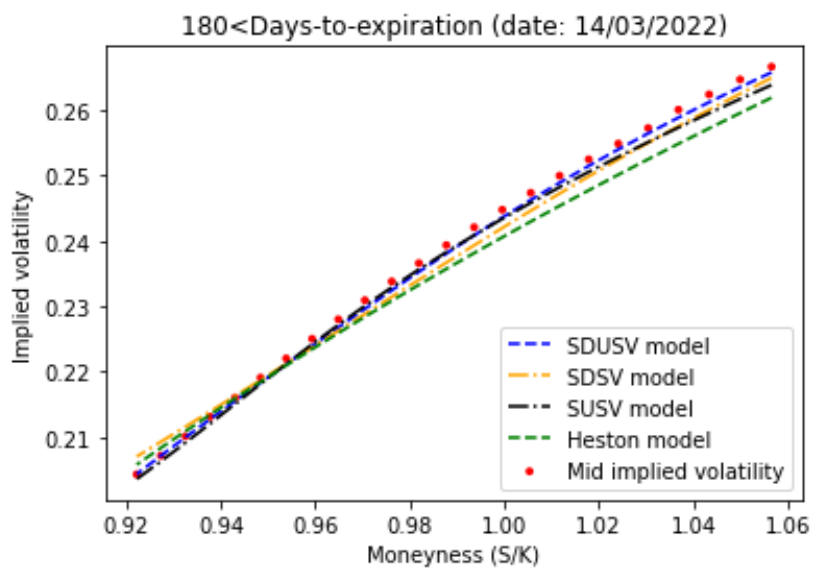}
			\end{minipage}%
		}%
  \centering
  \caption{Implied volatility curves of the four models and mid curves ($\tau$ is the normalized time to maturity measured in years)}
		\label{fig-imp}
	\end{figure}

We also examine the out-of-sample pricing performances of the four models. On each day, except the first day of calibration, we use previous day's option prices to calibrate model parameters and use them to calculate current day's option prices. Market variables such
as the index value $x$ and risk-free rate $r$ are updated as observed on the current day. The process is repeated until the last day of the calibration period. We need to point out that as for the SVSDU model, SVSD model and the SVSU model, we predict current day's option prices with previous day's drawdown ratio $\frac{x}{y}$ and drawup ratio $\frac{x}{z}$ instead of previous day's parameters $y$ and $z$. To be more specific, taking the SVSDU model as an example, on day $d$, after we back out the parameters $y_{d-1}$ and $z_{d-1}$ with previous day's option prices, we calculate the ratios $\frac{x_{d-1}}{y_{d-1}}$ and $\frac{x_{d-1}}{z_{d-1}}$, where $x_{d-1}$ is the index value on day $d-1$. Then we obtain $y_{new}=\frac{x_d y_{d-1}}{x_{d-1}}$ and $z_{new}=\frac{x_d z_{d-1}}{x_{d-1}}$, and use them to predict current day's option prices instead of using $y_{d-1}$ and $z_{d-1}$. If we directly use $y_{d-1}$ and $x_{d-1}$ to make a prediction, the ratios $\frac{x_{d}}{y_{d-1}}$ and $\frac{x_{d}}{z_{d-1}}$ are not consistent with stickiness coefficients $\xi$ and $\eta$ obtained from the previous day, since previous day's stickiness coefficients $\xi$ and $\eta$ measure the previous day's stickiness of drawdown ratio and drawup ratio at the level 1. The above procedure is the same for the SVSD model and the SVSU model when we make prediction. 

The average out-of-sample calibration errors are shown in Table \ref{tab-out-error}. We can find that the SVSDU model, SVSD model and the SVSU model beat the classical Heston model in all of the market situations out-of-sample. So the stickiness parameters improve models' structural fitting. It is surprising that although the SVSDU model on average calibrates the best in-sample in 2021, the SVSD model achieves better out-of-sample performances than the SVSDU model in 2021. This is likely because the coexistence of the drawdown and drawup stickiness factors interferes with the model's predictions in an upward trend. In growing economics, losing streaks are less likely to appear in the near future, so the drawup stickiness factor may counteract some of the positive effects of the drawdown stickiness factor on the model's prediction performances. For the similar reason, the SVSU model is surpassed by the SVSDU model out-of-sample in 2021, because the drawdown stickiness factor offsets some of the negative influences of the drawup stickiness factor in prediction in an upward trend. The SVSDU model achieves the best out-of-sample performance in 2022, when the market experienced frequent fluctuations with both rises and falls occurring frequently. This finding implies that the SVSDU model has a good predictive ability in volatile market conditions. 

\begin{table}[!h]
\begin{center}
\begin{tabular}{cccc}
\hline
Sampling period       & Model  & APE      & AAE    \\ \hline
04/01/2021-31/12/2021 & Heston & 1.3280\% & 5.5841 \\
                      & SVSD   & 1.2310\% & 5.1792 \\
                      & SVSU   & 1.3060\% & 5.4929 \\
                      & SVSDU  & 1.2740\% & 5.3596 \\
03/01/2022-30/12/2022 & Heston & 1.9870\% & 5.7397 \\
                      & SVSD   & 1.8670\% & 5.3917 \\
                      & SVSU   & 1.8100\% & 5.2261 \\
                      & SVSDU  & 1.7870\% & 5.1618 \\ \hline
\end{tabular}
\end{center}
\caption{Out-of-sample calibration errors for four models in three periods}
\label{tab-out-error}
\end{table}

In order to further explore the out-of-sample performances of the four models, we follow \cite{bakshi1997empirical} and carry out out-of-sample testing under different moneyness-maturity groups. The testing is performed in 2021, when the record highs of index value occurred the most frequently and the SVSD model achieves the best overall prediction performances. The pricing results are shown in Table \ref{tab-cate-aae} and Table \ref{tab-cate-ape}. Under the "All-Option-Based" group, the results are obtained using the parameters implied by all of the previous day's call option prices. Under the "Maturity-Based" group, the results are obtained using the parameters implied by previous day's option prices of a given time to maturity to price the current day's options of the same time to maturity. Under the ``Moneyness-Based" group, the results are obtained using the parameters implied by previous day's option prices of a given moneyness to price the current day's options of the same moneyness.  A call option is said to be at-the-money (ATM) if its moneyness $S/K\in (0.97,1.03)$; out-of-the-money (OTM) if $S/K\le 0.97$; and in-the-money (ITM) if $S/K\ge 1.03$. In terms of expiration, an option can be classified as short-term ($<60$ days); medium-term ($60-180$ days); and long-term ($>180$ days). As for the "All-Option-Based" group, both pricing error measures rank the SVSD model first in most of the categories, except that for a few categories either the call options have short maturities or are OTM. One possible explanation is that when we solve the optimization problem \eqref{eq:modified-optimization-problem} using all of the call option data from the previous day, the objective function in \eqref{eq:modified-optimization-problem} and our calibration algorithm are biased in favor of more expensive calls (i.e., medium-term and long-term calls, ATM and ITM calls), especially for the SVSD model. So we can find that the SVSD model achieves the best results in the medium-term and long-term categories except for the deep OTM case under the "All-Option-Based" group. Another possible reason for poor performances of the SVSD model in OTM categories is that the SVSD model is good at modeling the dynamic of asset value in an upward trend, when most of the options are ITM.

As for the "Maturity-Based" results, both pricing error measures rank the Heston model last and the SVSD model the first in most of the categories. The SVSD model performs better under the "Maturity-Based" treatment for ITM and ATM calls with short maturities than under "All-Option-Based" treatment, with both pricing error measures ranking the SVSD model first in those categories. In addition, the SVSD model still achieves the best prediction results in medium-term category except for the deep OTM case. The reason is that the trend of the index value tends to persist in the short and medium term. If the index value is in an upward trend, then in the short or medium term, the index value is highly likely to continue to rise. So the SVSD model can offer more accurate prediction of option prices. The SVSDU model improves the most in pricing long-term options. This is likely due to the fact that the trend of the index value may rise or fall in the long run, and the SVSDU model with both drawdown stickiness parameter and drawup stickiness parameter is able to better reflect the effect of long-term uncertainty on option values.  

As for the results from the "Moneyness-Based" group, both pricing error measures rank the Heston model last and the SVSD model the first except for a few categories. The SVSD model benefits the most from the "Moneyness-Based" treatment during 2021. Regardless of maturity, both pricing error measures rank the SVSD model the first except for the deep OTM case, while the SVSDU model achieves the best performance in predicting deep OTM call option prices under "Moneyness-Based" treatment.

\begin{table}[!h]
\resizebox{\linewidth}{!}{
\begin{tabular}{ccccclccclccc}
\hline
\multicolumn{1}{l}{}                                      & \multicolumn{1}{l}{} & \multicolumn{3}{c}{\begin{tabular}[c]{@{}c@{}}All-Option-Based\\ Days-to-Expiration\end{tabular}} &  & \multicolumn{3}{c}{\begin{tabular}[c]{@{}c@{}}Maturity-Based\\ Days-to-Expiration\end{tabular}} &  & \multicolumn{3}{c}{\begin{tabular}[c]{@{}c@{}}Moneyness-Bsed\\ Days-to-Expiration\end{tabular}} \\ \cline{3-5} \cline{7-9} \cline{11-13} 
\begin{tabular}[c]{@{}c@{}}Moneyness\\ ($S/K$)\end{tabular} & Model                & \multicolumn{1}{l}{\textless{}60}  & \multicolumn{1}{l}{60-180}  & \multicolumn{1}{l}{$\ge$ 180}  &  & \multicolumn{1}{l}{\textless{}60}  & \multicolumn{1}{l}{60-180} & \multicolumn{1}{l}{$\ge$ 180} &  & \multicolumn{1}{l}{\textless{}60}  & \multicolumn{1}{l}{60-180} & \multicolumn{1}{l}{$\ge$ 180} \\ \hline
\textless{}0.94                                           & Heston               & 3.53                               & 4.11                        & 6.33                           &  & 3.44                               & 3.27                       & 4.56                          &  & 3.70                               & 2.89                       & 3.75                          \\
& SVSD                 & 4.74                               & 5.18                        & 7.11                           &  & 2.98                               & 3.83                       & 6.38                          &  & 1.39                               & 2.27                       & 3.27                          \\
& SVSU                 & 2.97                               & 5.01                        & 9.78                           &  & 3.03                               & 3.28                       & 6.64                          &  & 1.50                               & 2.34                       & 3.34                          \\
& SVSDU                & 3.14                               & 4.78                        & 6.35                           &  & 2.47                               & 3.09                       & 3.65                          &  & 1.25                               & 2.26                       & 3.13                          \\
0.94-0.97                                                 & Heston               & 3.95                               & 5.61                        & 11.22                           &  & 3.36                               & 5.06                       & 7.45                          &  & 3.33                               & 4.79                       & 6.32                          \\
& SVSD                 & 3.53                               & 5.24                        & 7.91                           &  & 3.39                               & 4.80                       & 7.31                          &  & 2.00                               & 4.14                       & 4.98                          \\
& SVSU                 & 3.31                               & 5.26                        & 8.07                           &  & 3.37                               & 5.11                       & 6.63                          &  & 2.33                               & 4.75                       & 5.78                          \\
& SVSDU                & 3.82                               & 5.37                        & 8.43                          &  & 2.99                               & 5.07                       & 6.17                          &  & 2.37                               & 4.93                       & 6.32                          \\
0.97-1.00                                                 & Heston               & 5.83                               & 6.47                        & 9.31                           &  & 5.11                               & 6.84                       & 7.44                          &  & 5.21                               & 6.72                       & 7.30                          \\
& SVSD                 & 5.84                               & 6.20                        & 7.34                           &  & 5.10                               & 6.32                       & 7.34                          &  & 4.51                               & 5.83                       & 5.77                          \\
& SVSU                 & 6.21                               & 7.12                        & 7.53                           &  & 5.38                               & 6.59                       & 6.71                          &  & 4.97                               & 6.70                       & 6.52                          \\
& SVSDU                & 6.22                               & 6.65                        & 8.23                           &  & 5.13                               & 6.62                      & 6.82                          &  & 4.84                               & 6.53                       & 6.79                          \\
1.00-1.03                                                 & Heston               & 6.50                               & 6.80                        & 7.33                           &  & 6.06                               & 7.01                       & 6.84                          &  & 6.22                               & 7.10                       & 7.43                          \\
& SVSD                 & 6.09                               & 6.61                        & 6.77                           &  & 5.85                               & 6.62                       & 6.96                          &  & 5.57                               & 6.35                       & 5.92                          \\
& SVSU                 & 6.81                               & 7.73                        & 7.16                           &  & 6.12                               & 7.10                      & 6.94                          &  & 6.25                               & 7.23                       & 6.85                          \\
& SVSDU                & 6.37                               & 7.18                        & 7.78                           &  & 5.99                               & 6.96                       & 7.00                          &  & 5.88                               & 6.88                       & 6.64                          \\
1.03-1.06                                                 & Heston               & 5.28                               & 6.93                        & 7.17                           &  & 5.20                               & 6.82                       & 7.00                          &  & 5.37                               & 7.08                       & 7.67                          \\
& SVSD                 & 5.89                               & 6.52                        & 6.43                           &  &4.95                               & 6.43                       & 6.51                          &  & 4.78                               & 6.22                       & 5.85                          \\
& SVSU                 & 6.01                               & 7.19                        & 6.97                           &  &5.15                               & 6.90                       & 6.92                          &  & 5.11                               & 6.73                       & 6.77                          \\
& SVSDU                & 5.29                               & 6.94                        & 7.45                           &  & 4.95                               & 6.67                       & 6.91                          &  & 4.82                               & 6.58                       & 6.48                          \\
$\ge$1.06                                                 & Heston               & 4.18                               & 5.65                        & 8.00                           &  & 4.10                               & 5.53                       & 7.72                          &  & 4.07                               & 5.22                       & 6.68                          \\
& SVSD                 & 4.27                               & 4.58                        & 5.14                           &  & 3.87                               & 4.59                       & 5.54                          &  & 3.97                               & 4.31                       & 4.14                          \\
& SVSU                 & 4.15                               & 5.01                        & 5.42                           &  & 3.88                               & 4.73                       & 5.35                          &  & 4.20                               & 4.86                      & 5.59                          \\
& SVSDU                & 4.63                               &4.73                        & 5.82                           &  & 4.53                               & 4.74                       & 5.08                          &  & 4.33                               & 4.84                       & 5.24                          \\ \hline
\end{tabular}}
\caption{Out-of-sample average absolute errors (AAE) between the market prices and the model prices for each call in a given moneyness–maturity category in 2021}
\label{tab-cate-aae}
\end{table}

\begin{table}[!h]
\resizebox{\linewidth}{!}{
\begin{tabular}{ccccclccclccc}
\hline
\multicolumn{1}{l}{}                                      & \multicolumn{1}{l}{} & \multicolumn{3}{c}{\begin{tabular}[c]{@{}c@{}}All-Option-Based\\ Days-to-Expiration\end{tabular}} &  & \multicolumn{3}{c}{\begin{tabular}[c]{@{}c@{}}Maturity-Based\\ Days-to-Expiration\end{tabular}} &  & \multicolumn{3}{c}{\begin{tabular}[c]{@{}c@{}}Moneyness-Bsed\\ Days-to-Expiration\end{tabular}} \\ \cline{3-5} \cline{7-9} \cline{11-13} 
\begin{tabular}[c]{@{}c@{}}Moneyness\\ ($S/K$)\end{tabular} & Model                & \multicolumn{1}{l}{\textless{}60}  & \multicolumn{1}{l}{60-180}  & \multicolumn{1}{l}{$\ge$ 180}  &  & \multicolumn{1}{l}{\textless{}60}  & \multicolumn{1}{l}{60-180} & \multicolumn{1}{l}{$\ge$ 180} &  & \multicolumn{1}{l}{\textless{}60}  & \multicolumn{1}{l}{60-180} & \multicolumn{1}{l}{$\ge$ 180} \\ \hline
\textless{}0.94                                           & Heston               & 94.67\%                           & 21.67\%                     & 15.30\%                        &  & 92.10\%                           & 17.23\%                    & 11.01\%                       &  & 98.98\%                           & 15.20\%                    & 9.06\%                        \\
& SVSD                 & 127.11\%                            & 27.30\%                     & 17.19\%                        &  & 79.91\%                            & 20.17\%                    & 15.41\%                       &  & 37.18\%                            & 11.97\%                    & 7.90\%                        \\
& SVSU                 & 79.63\%                            & 26.39\%                     & 23.63\%                        &  & 81.20\%                            & 17.28\%                    & 16.04\%                       &  &40.23\%                            & 12.33\%                    & 8.08\%                        \\
& SVSDU                & 84.25\%                            & 25.16\%                     & 15.35\%                        &  &66.33\%                            & 16.28\%                    & 8.83\%                        &  & 33.47\%                            & 11.90\%                    & 7.58\%                        \\
0.94-0.97                                                 & Heston               & 41.05\%                            & 9.46\%                      & 7.50\%                         &  & 34.89\%                            & 8.53\%                    & 4.98\%                        &  & 34.57\%                            & 8.07\%                     & 4.23\%                        \\
& SVSD                 & 36.74\%                            & 8.84\%                      & 5.29\%                         &  & 35.24\%                            & 8.09\%                     & 4.88\%                        &  & 20.81\%                            & 6.98\%                     & 3.33\%                        \\
& SVSU                 & 34.36\%                            & 8.87\%                      & 5.39\%                         &  &35.07\%                            & 8.61\%                     & 4.43\%                        &  & 24.18\%                            & 8.00\%                     & 3.86\%                        \\
& SVSDU                & 39.78\%                            & 9.05\%                      &5.63\%                         &  &31.10\%                            & 8.54\%                     &4.12\%                        &  &24.63\%                            &8.31\%                     &4.23\%                        \\
0.97-1.00                                                 & Heston               &17.64\%                            &5.45\%                      &4.21\%                         &  &15.46\%                             & 5.77\%                     &3.37\%                        &  &15.74\%                             &5.66\%                     &3.30\%                        \\
& SVSD                 & 17.68\%                            & 5.22\%                      &3.32\%                         &  &15.44\%                             & 5.33\%                     &3.36\%                        &  &13.65\%                             &4.92\%                     &2.61\%                        \\
& SVSU                 &18.78\%                            &6.00\%                      &3.40\%                         &  &16.29\%                             & 5.55\%                     &3.04\%                        &  &15.03\%                             &5.64\%                     &2.95\%                        \\
& SVSDU                &18.83\%                             &5.61\%                      &3.72\%                         &  &15.53\%                             & 5.58\%                     &3.08\%                        &  &14.63\%                             &5.50\%                     &3.07\%                        \\
1.00-1.03                                                 & Heston               &6.10\%                             &3.43\%                      &2.43\%                         &  &5.69\%                             & 3.54\%                     &2.27\%                        &  &5.84\%                             &3.58\%                     &2.47\%                        \\
& SVSD                 & 5.72\%                             & 3.34\%                      & 2.25\%                         &  &5.49\%                             & 3.34\%                     & 2.31\%                        &  & 5.23\%                             & 3.20\%                     & 1.97\%                        \\
& SVSU                 & 6.39\%                             & 3.90\%                      & 2.38\%                         &  & 5.74\%                             & 3.58\%                     & 2.30\%                        &  & 5.87\%                             & 3.65\%                     &2.28\%                        \\
& SVSDU                & 5.98\%                             & 3.62\%                      & 2.58\%                         &  &5.62\%                             & 3.51\%                     & 2.32\%                        &  & 5.52\%                             &3.47\%                     &2.21\%                        \\
1.03-1.06                                                 & Heston               &2.54\%                             &2.41\%                      &1.87\%                         &  &2.51\%                             & 2.37\%                     &1.82\%                        &  &2.59\%                             &2.47\%                     &2.00\%                        \\
& SVSD                 & 2.83\%                             & 2.27\%                      &1.67\%                         &  &2.38\%                             & 2.24\%                     &1.69\%                        &  &2.30\%                             &2.17\%                     &1.52\%                        \\
& SVSU                 &2.89\%                             &2.50\%                      &1.81\%                         &  &2.48\%                             & 2.40\%                     & 1.80\%                        &  &2.45\%                             &2.35\%                     &1.76\%                        \\
& SVSDU                &2.55\%                             &2.41\%                      &1.94\%                         &  &2.38\%                             & 2.32\%                     &1.79\%                        &  &2.32\%                             &2.29\%                     &1.69\%                        \\
$\ge$1.06                                                 & Heston               &0.64\%                             & 0.76\%                      & 0.84\%                         &  & 0.63\%                             & 0.75\%                     & 0.81\%                        &  & 0.63\%                             &0.71\%                     &0.70\%                        \\
& SVSD                 &0.66\%                             &0.61\%                      &0.54\%                         &  &0.59\%                             & 0.62\%                     &0.58\%                        &  &0.61\%                             &0.58\%                     &0.43\%                        \\
& SVSU                 &0.63\%                             & 0.67\%                      & 0.56\%                         &  &0.60\%                             & 0.63\%                     & 0.56\%                        &  & 0.65\%                             & 0.66\%                     &0.58\%                        \\
& SVSDU                & 0.71\%                             &0.63\%                      & 0.60\%                         &  &0.69\%                             & 0.64\%                     & 0.53\%                        &  & 0.66\%                             & 0.65\%                     & 0.55\%                        \\ \hline
\end{tabular}}
\caption{Out-of-sample average percentage errors (APE) between market prices and the model prices for each call in a given moneyness-maturity category in 2021}
\label{tab-cate-ape}
\end{table}

\section{Conclusion}
\label{sec-con}
We have developed a novel financial model, the SVSDU model that admits drawdown stickiness factor and drawup stickiness factor, which can explain the notable features of winning and losing streak that can not be captured simultaneously by existing quantitative models in the financial market. Besides, the incorporation of stochastic volatility in our model helps account for other stylized facts of market data. Due to the lack of closed-form option pricing formula under the SVSDU model, we use a deep neural network to approximate the option pricing formula by solving the corresponding pricing PDE. The numerical experiments demonstrate the accuracy and efficiency of our deep learning method. In order to calibrate the SVSDU model to historical data, we develop a novel calibration framework for the neural network approximation. We compare calibration results of the SVSDU model with other three financial models: the SVSD model, the SVSU model and the Heston model. We apply the four models to SPX European call option prices in 2021 when the winning streak appeared most of time, and in 2022 when both the winning and losing streak appeared frequently. The empirical studies show that the SVSDU model performs the best in-sample in all of the periods and out-of-sample in the volatile market, while the SVSD model (SVSU model) performs well out-of-sample in winning streak period (losing streak period). Those facts imply that the SVSDU model can reflect the effects of frequent persistent extremes (maxima and minima) on option values pretty well and it is a good reflection of economic reality, since both rise and fall in asset prices are common in the financial market. In addition, in a concentrated period of time, the SVSD model is good at predicting the option values when the underlying asset value dynamics experience a continuous growth, while the SVSU model makes a good prediction on the option prices when the underlying asset value dynamics undergo a consistent decline.

In the future research, we are interested in the performances of calibrated drawdown and drawup stickiness factors in portfolio optimization, since the stickiness factors evaluate the effects of extreme persistence of the underlying asset values.

\appendices
\section{Proof}
\label{app-proof}

\begin{proof}[Proof of the Theorem \ref{thm:existence-weak-sol-SDE}]

    Let $\mathbb{S}=\left\{\boldsymbol{x} \in \mathbb{R}^3: \Phi(\boldsymbol{x})>0\right\} $ and $ \partial \mathbb{S}=\left\{\boldsymbol{x} \in \mathbb{R}^3: \Phi(\boldsymbol{x})=0\right\}$, where $\Phi(\boldsymbol{x})=\prod_{i=1}^{\hat{d}} (1-e^{-x_i})$ and $\hat{d}=2$ is the number of dimensions exhibiting stickiness. We set $\widetilde{U}^{\pm}_t=\ln U_t^{\pm}$ and  $\widetilde{D}^{\pm}_t=-\ln D_t^{\pm}$. The SDE \eqref{eq:sticky-process-SDE} can be rewritten as
    \begin{align}\label{eq:ln-sticky-process-SDE}
	\begin{cases}
	&d\widetilde{D}^{\pm}_t = \mathbf{1}_{\{ \widetilde{D}^{\pm}_t>0, \widetilde{U}^{\pm}_t>0 \}}\big((r-\mu+\frac{1}{2}V_t) dt - \sqrt{V_t} dB_t \big)  + dL_t^0(\widetilde{D}^\pm),\\
	&d\widetilde{U}^\pm_t = \mathbf{1}_{\{ \widetilde{D}^{\pm}_t>0, \widetilde{U}^{\pm}_t>0 \}}\big((\mu-r-\frac{1}{2}V_t)dt + \sqrt{V_t} dB_t \big)  + d L_t^0(\widetilde{U}^\pm),\\
	&dV_t=\kappa(\theta-V_t)dt+\sigma\sqrt{V_t}dW_t,\\
	&\int_{0}^{t} 1_{\{\widetilde{D}^{\pm}_s = 0\}} ds  =  \xi L_t^0(\widetilde{D}^{\pm}),\\
	&\int_{0}^{t} 1_{\{\widetilde{U}^\pm_s = 0\}} ds  =  \eta L_t^0(\widetilde{U}^\pm).
	\end{cases}	
	\end{align}

Since the Feller condition $2\kappa\theta>\sigma^2$ is satisfied, $V_t>0$. We set $P_t^{\pm}=\frac{\widetilde{U}^\pm_t}{V_t+1}$ and $Q_t^\pm=\frac{\widetilde{D}^{\pm}_t}{V_t+1}$. So
\begin{align}
    dP_t^{\pm} &= \widetilde{U}^\pm_t d\Big(\frac{1}{V_t+1}\Big)+\frac{1}{V_t+1}d\widetilde{U}^\pm_t+d\widetilde{U}^\pm_t d\Big(\frac{1}{V_t+1}\Big)\\
  &=\Big(-\frac{P^{\pm}_t}{V_t+1} k\big(\theta-V_t\big)+\frac{P^{\pm}_t}{\left(V_t+1\right)^2} \sigma^2 V_t\Big) d t-\frac{P^{\pm}_t \sigma \sqrt{V_t}}{V_t+1} d W_t \\ &\quad\ +\mathbf{1}_{\{P^{\pm}_t>0, Q^{\pm}_t>0\}}\Big(\big(\frac{\mu-r-\frac{1}{2} V_t}{V_t+1}-\frac{1}{\left(V_t+1\right)^2} \sigma \rho V_t \big)d t+\frac{\sqrt{V_t}}{V_t+1} d B_t\Big) \\ &\quad\ +\frac{1}{\left(V_t+1\right) \eta} \mathbf{1}_{\{P^{\pm}_t=0\}} d t.
\end{align}
Let $\widetilde{P}_t^\pm=\ln (P^\pm_t+1)$, $\widetilde{Q}_t^\pm=\ln (Q^\pm_t+1)$, and $\widetilde{V}_t=\ln (V_t+1)$. Define $f(x)=e^x-1$. Then we have
\begin{align}
    V_t=f(\widetilde{V}_t)>0, \ P^\pm_t=f(\widetilde{P}^{\pm}_t)\ge 0, \ Q^\pm_t=f(\widetilde{Q}^{\pm}_t)\ge 0.
\end{align}

$\widetilde{V}_t$ satisfies
\begin{align}
    d\widetilde{V}_t&=\frac{1}{V_t+1}dV_t-\frac{1}{(V_t+1)^2}d\left\langle V,V  \right\rangle_t\\
    &=\Big(\frac{\kappa(\theta-V_t)}{V_t+1}-\frac{\sigma^2V_t}{(V_t+1)^2}  \Big)dt+\frac{\sigma\sqrt{V_t}}{V_t+1}dW_t\\
    &=\Big(\frac{\kappa(\theta-f(\widetilde{V_t}))}{f(\widetilde{V_t})+1}-\frac{\sigma^2f(\widetilde{V_t})}{(f(\widetilde{V_t})+1)^2}  \Big)dt+\frac{\sigma\sqrt{f(\widetilde{V_t})}}{f(\widetilde{V_t})+1}dW_t.
\end{align}

$\widetilde{P}_t^\pm$ satisfies 
\begin{align}
    d\widetilde{P}^{\pm}_t&=\frac{1}{P^{\pm}_t+1}dP^{\pm}_t-\frac{1}{(P^{\pm}_t+1)^2}d\left\langle P^{\pm},P^{\pm}\right\rangle_t\\
    &=\bigg(\Big(-\frac{\kappa(\theta-f(\widetilde{V_t}))f(\widetilde{P}^{\pm}_t)}{(f(\widetilde{V_t})+1)(f(\widetilde{P}^{\pm}_t)+1)}+\frac{\sigma^2f(\widetilde{V}_t)f(\widetilde{P}^{\pm}_t)}{(f(\widetilde{V}_t)+1)^2(f(\widetilde{P}^{\pm}_t)+1)}+\frac{\mu-r-\frac{1}{2}f(\widetilde{V}_t)}{(f(\widetilde{P}^{\pm}_t)+1)(f(\widetilde{V}_t)+1)} \\
    &\quad \ -\frac{\sigma\rho f(\widetilde{V}_t)}{(f(\widetilde{P}^{\pm}_t)+1)(f(\widetilde{V}_t)+1)^2}-\frac{\sigma^2f(\widetilde{P}^{\pm}_t)^2f(\widetilde{V}_t)}{(f(\widetilde{P}^{\pm}_t)+1)^2(f(\widetilde{V}_t)+1)^2}-\frac{f(\widetilde{V}_t)}{(f(\widetilde{P}^{\pm}_t)+1)^2(f(\widetilde{V}_t)+1)^2}\\
    &\quad \ +\frac{2\rho\sigma f(\widetilde{P}^{\pm}_t)f(\widetilde{V}_t)}{(f(\widetilde{P}^{\pm}_t)+1)^2(f(\widetilde{V}_t)+1)^2}\Big)dt -\frac{\sigma f(\widetilde{P}^{\pm}_t)\sqrt{f(\widetilde{V}_t)}}{(f(\widetilde{P}^{\pm}_t)+1)(f(\widetilde{V}_t)+1)}dW_t\\
    &\quad \ +\frac{\sqrt{f(\widetilde{V}_t)}}{(f(\widetilde{P}^{\pm}_t)+1)(f(\widetilde{V}_t)+1)}dB_t \bigg)\mathbf{1}_{\{\widetilde{P}^{\pm}_t>0, \widetilde{Q}^{\pm}_t>0\}}
    +\frac{1}{(f(\widetilde{V}_t)+1)\eta}\mathbf{1}_{\{\widetilde{P}^{\pm}_t=0\}}dt \\
    &= \mathbf{1}_{\{\widetilde{P}^{\pm}_t>0, \widetilde{Q}^{\pm}_t>0\}}\bigg(\Big(\frac{-\kappa(\theta-f(\widetilde{V_t}))f(\widetilde{P}^{\pm}_t)+\mu-r-\frac{1}{2}f(\widetilde{V}_t)}{(f(\widetilde{V_t})+1)(f(\widetilde{P}^{\pm}_t)+1)}+\frac{\sigma^2f(\widetilde{V}_t)f(\widetilde{P}^{\pm}_t)-\sigma\rho f(\widetilde{V}_t)}{(f(\widetilde{V}_t)+1)^2(f(\widetilde{P}^{\pm}_t)+1)}\\
    &\quad \ +\frac{-\sigma^2f(\widetilde{P}^{\pm}_t)^2f(\widetilde{V}_t)-f(\widetilde{V}_t)+2\rho\sigma f(\widetilde{P}^{\pm}_t)f(\widetilde{V}_t)}{(f(\widetilde{P}^{\pm}_t)+1)^2(f(\widetilde{V}_t)+1)^2}\Big)dt-\frac{\sigma f(\widetilde{P}^{\pm}_t)\sqrt{f(\widetilde{V}_t)}}{(f(\widetilde{P}^{\pm}_t)+1)(f(\widetilde{V}_t)+1)}dW_t\\
    &\quad \ +\frac{\sqrt{f(\widetilde{V}_t)}}{(f(\widetilde{P}^{\pm}_t)+1)(f(\widetilde{V}_t)+1)}dB_t\bigg)+\mathbf{1}_{\{\widetilde{P}^{\pm}_t=0\}} \frac{1}{(f(\widetilde{V}_t)+1)\eta}dt.
\end{align}
Similarly, we have
\begin{align}
    d\widetilde{Q}_t^{\pm}&=\mathbf{1}_{\{\widetilde{P}^{\pm}_t>0, \widetilde{Q}^{\pm}_t>0\}}\bigg(\Big(\frac{r-\mu-\frac{1}{2}f(\widetilde{V}_t)-f(\widetilde{Q}^{\pm}_t)\kappa(\theta-f(\widetilde{V}_t))}{(f(\widetilde{Q}_t^\pm)+1)(f(\widetilde{V}_t)+1)}+\frac{\sigma\rho f(\widetilde{V}_t)+\sigma^2 f(\widetilde{Q}_t^\pm)f(\widetilde{V}_t)}{(f(\widetilde{Q}_t^\pm)+1)(f(\widetilde{V}_t)+1)^2}\\
    &\quad\ - \frac{\sigma^2f(\widetilde{Q}_t^\pm)^2f(\widetilde{V}_t)+f(\widetilde{V}_t)+2\sigma\rho f(\widetilde{Q}_t^\pm)f(\widetilde{V}_t)}{(f(\widetilde{Q}_t^\pm)+1)^2(f(\widetilde{V}_t)+1)^2}\Big)dt-\frac{\sigma f(\widetilde{Q}_t^\pm)\sqrt{f(\widetilde{V}_t)}}{(f(\widetilde{Q}_t^\pm)+1)(f(\widetilde{V}_t)+1)}dW_t\\
    &\quad\ -\frac{\sqrt{f(\widetilde{V}_t)}}{(f(\widetilde{Q}_t^\pm)+1)(f(\widetilde{V}_t)+1)}dB_t\bigg)+\mathbf{1}_{\{\widetilde{Q}_t^\pm=0\}}\frac{1}{(f(\widetilde{V}_t)+1)\xi}dt.
\end{align}

Note that $dW_t=\rho dB_t+\sqrt{1-\rho^2}d\bar{B}_t$ where $\{B_t,t\ge 0\}$ and $\{\bar{B}_t,t\ge 0\}$ are two independent standard Brownian motion. After the successive transformations, we obtain the SDE system of $\boldsymbol{X}_t=(\widetilde{P}_t^\pm, \widetilde{Q}_t^\pm, \widetilde{V}_t)^\top$ from the SDE \eqref{eq:sticky-process-SDE}:
\begin{align}\label{eq:final-sticky-process-SDE}
	\begin{cases}
	&d\widetilde{Q}_t^\pm=\mathbf{1}_{\{\boldsymbol{X}_t\in \mathbb{S}\}}\bigg(\Big(\frac{r-\mu-\frac{1}{2}f(\widetilde{V}_t)-f(\widetilde{Q}^{\pm}_t)\kappa(\theta-f(\widetilde{V}_t))}{(f(\widetilde{Q}_t^\pm)+1)(f(\widetilde{V}_t)+1)}+\frac{\sigma\rho f(\widetilde{V}_t)+\sigma^2 f(\widetilde{Q}_t^\pm)f(\widetilde{V}_t)}{(f(\widetilde{Q}_t^\pm)+1)(f(\widetilde{V}_t)+1)^2}\\
    &\quad\quad \quad \    - \frac{\sigma^2f(\widetilde{Q}_t^\pm)^2f(\widetilde{V}_t)+f(\widetilde{V}_t)+2\sigma\rho f(\widetilde{Q}_t^\pm)f(\widetilde{V}_t)}{(f(\widetilde{Q}_t^\pm)+1)^2(f(\widetilde{V}_t)+1)^2}\Big)dt-\frac{\sigma f(\widetilde{Q}_t^\pm)\sqrt{f(\widetilde{V}_t)}\sqrt{1-\rho^2}}{(f(\widetilde{Q}_t^\pm)+1)(f(\widetilde{V}_t)+1)}d\bar{B}_t\\
    &\quad\quad \quad \ -\Big(\frac{\sqrt{f(\widetilde{V}_t)}}{(f(\widetilde{Q}_t^\pm)+1)(f(\widetilde{V}_t)+1)}+\frac{\sigma f(\widetilde{Q}_t^\pm)\sqrt{f(\widetilde{V}_t)}\rho}{(f(\widetilde{Q}_t^\pm)+1)(f(\widetilde{V}_t)+1)}\Big)dB_t\bigg)\\
    &\quad\quad \quad \ + \mathbf{1}_{\{\boldsymbol{X}_t\in \partial \mathbb{S}\}}\mathbf{1}_{\{\widetilde{Q}_t^\pm=0\}}\frac{1}{(f(\widetilde{V}_t)+1)\xi}dt,\\
    &d\widetilde{P}_t^\pm= \mathbf{1}_{\{\boldsymbol{X}_t\in \mathbb{S}\}}\bigg(\Big(\frac{-\kappa(\theta-f(\widetilde{V_t}))f(\widetilde{P}^{\pm}_t)+\mu-r-\frac{1}{2}f(\widetilde{V}_t)}{(f(\widetilde{V_t})+1)(f(\widetilde{P}^{\pm}_t)+1)}+\frac{\sigma^2f(\widetilde{V}_t)f(\widetilde{P}^{\pm}_t)-\sigma\rho f(\widetilde{V}_t)}{(f(\widetilde{V}_t)+1)^2(f(\widetilde{P}^{\pm}_t)+1)}\\
    &\quad\quad \quad \  +\frac{-\sigma^2f(\widetilde{P}^{\pm}_t)^2f(\widetilde{V}_t)-f(\widetilde{V}_t)+2\rho\sigma f(\widetilde{P}^{\pm}_t)f(\widetilde{V}_t)}{(f(\widetilde{P}^{\pm}_t)+1)^2(f(\widetilde{V}_t)+1)^2}\Big)dt-\frac{\sigma f(\widetilde{P}^{\pm}_t)\sqrt{f(\widetilde{V}_t)}\sqrt{1-\rho^2}}{(f(\widetilde{P}^{\pm}_t)+1)(f(\widetilde{V}_t)+1)}d\bar{B}_t\\
    &\quad\quad \quad \  +\Big(\frac{\sqrt{f(\widetilde{V}_t)}}{(f(\widetilde{P}^{\pm}_t)+1)(f(\widetilde{V}_t)+1)}-\frac{\sigma f(\widetilde{P}^{\pm}_t)\sqrt{f(\widetilde{V}_t)}\rho}{(f(\widetilde{P}^{\pm}_t)+1)(f(\widetilde{V}_t)+1)}\Big)dB_t\bigg)\\
    &\quad\quad \quad \ +\mathbf{1}_{\{\boldsymbol{X}_t\in \partial \mathbb{S}\}}\mathbf{1}_{\{\widetilde{P}^{\pm}_t=0\}} \frac{1}{(f(\widetilde{V}_t)+1)\eta}dt,\\
	&d\widetilde{V}_t=\mathbf{1}_{\{\boldsymbol{X}_t\in \mathbb{S}\}}\bigg(\Big(\frac{\kappa(\theta-f(\widetilde{V_t}))}{f(\widetilde{V_t})+1}-\frac{\sigma^2f(\widetilde{V_t})}{(f(\widetilde{V_t})+1)^2}  \Big)dt+\frac{\sigma\rho\sqrt{f(\widetilde{V_t})}}{f(\widetilde{V_t})+1}dB_t+\frac{\sigma\sqrt{1-\rho^2}\sqrt{f(\widetilde{V_t})}}{f(\widetilde{V_t})+1}d\bar{B}_t\bigg)\\
    &\quad\quad \ \ +\mathbf{1}_{\{\boldsymbol{X}_t\in \partial \mathbb{S}\}}\bigg(\Big(\frac{\kappa(\theta-f(\widetilde{V_t}))}{f(\widetilde{V_t})+1}-\frac{\sigma^2f(\widetilde{V_t})}{(f(\widetilde{V_t})+1)^2}  \Big)dt+\frac{\sigma\rho\sqrt{f(\widetilde{V_t})}}{f(\widetilde{V_t})+1}dB_t+\frac{\sigma\sqrt{1-\rho^2}\sqrt{f(\widetilde{V_t})}}{f(\widetilde{V_t})+1}d\bar{B}_t\bigg),\\
    &\int_{0}^{t} \mathbf{1}_{\{\widetilde{Q}^\pm_s = 0\}} ds  =  \xi L_t^0(\widetilde{Q}^\pm)\\
        &\int_{0}^{t} \mathbf{1}_{\{\widetilde{P}^\pm_s = 0\}} ds  =  \eta L_t^0(\widetilde{P}^\pm).
	\end{cases}	
	\end{align}
The SDE \eqref{eq:final-sticky-process-SDE} can be formulated as
\begin{align}\label{eq:multidimensional-SDE-formulation}
	d \boldsymbol{X}_t= & \boldsymbol{\mu}\left(\boldsymbol{X}_t\right) \mathbf{1}\left(\boldsymbol{X}_t \in \mathbb{S}\right) d t+\boldsymbol{\Sigma}\left(\boldsymbol{X}_t\right) \mathbf{1}\left(\boldsymbol{X}_t \in \mathbb{S}\right) d \mathbf{B}_{1, t} \\ & +\hat{\boldsymbol{\beta}}\left(\boldsymbol{X}_t\right) \mathbf{1}\left(\boldsymbol{X}_t \in \partial \mathbb{S}\right) d t+\hat{\boldsymbol{\Gamma}}\left(\boldsymbol{X}_t\right) \mathbf{1}\left(\boldsymbol{X}_t \in \partial \mathbb{S}\right) d \boldsymbol{B}_{2, t} 
\end{align}
with the following sojourn condition
\begin{align}
    \mathbf{1}_{\{\boldsymbol{X}_t\in \partial \mathbb{S}\}}dt=\rho (\boldsymbol{X}_t)dL_t,
\end{align}
where $L$ is the local time process of $\boldsymbol{X}$ on the boundary. In \eqref{eq:multidimensional-SDE-formulation}, 
\begin{align}
\boldsymbol{\mu}\left(\boldsymbol{X}_t\right)=\left[\begin{array}{c}\frac{r-\mu-\frac{1}{2}f(\widetilde{V}_t)-f(\widetilde{Q}^{\pm}_t)\kappa(\theta-f(\widetilde{V}_t))}{(f(\widetilde{Q}_t^\pm)+1)(f(\widetilde{V}_t)+1)}+\frac{\sigma\rho f(\widetilde{V}_t)+\sigma^2 f(\widetilde{Q}_t^\pm)f(\widetilde{V}_t)}{(f(\widetilde{Q}_t^\pm)+1)(f(\widetilde{V}_t)+1)^2}- \frac{\sigma^2f(\widetilde{Q}_t^\pm)^2f(\widetilde{V}_t)+f(\widetilde{V}_t)+2\sigma\rho f(\widetilde{Q}_t^\pm)f(\widetilde{V}_t)}{(f(\widetilde{Q}_t^\pm)+1)^2(f(\widetilde{V}_t)+1)^2} \\ \frac{-\kappa(\theta-f(\widetilde{V_t}))f(\widetilde{P}^{\pm}_t)+\mu-r-\frac{1}{2}f(\widetilde{V}_t)}{(f(\widetilde{V_t})+1)(f(\widetilde{P}^{\pm}_t)+1)}+\frac{\sigma^2f(\widetilde{V}_t)f(\widetilde{P}^{\pm}_t)-\sigma\rho f(\widetilde{V}_t)}{(f(\widetilde{V}_t)+1)^2(f(\widetilde{P}^{\pm}_t)+1)} +\frac{-\sigma^2f(\widetilde{P}^{\pm}_t)^2f(\widetilde{V}_t)-f(\widetilde{V}_t)+2\rho\sigma f(\widetilde{P}^{\pm}_t)f(\widetilde{V}_t)}{(f(\widetilde{P}^{\pm}_t)+1)^2(f(\widetilde{V}_t)+1)^2} \\ \frac{\kappa(\theta-f(\widetilde{V_t}))}{f(\widetilde{V_t})+1}-\frac{\sigma^2f(\widetilde{V_t})}{(f(\widetilde{V_t})+1)^2}  \end{array}\right],
\end{align}
\begin{align}
\boldsymbol{\Sigma}\left(\boldsymbol{X}_t\right)=\left[\begin{array}{ccc}-\frac{\sigma f(\widetilde{Q}_t^\pm)\sqrt{f(\widetilde{V}_t)}\sqrt{1-\rho^2}}{(f(\widetilde{Q}_t^\pm)+1)(f(\widetilde{V}_t)+1)} & -\Big(\frac{\sqrt{f(\widetilde{V}_t)}}{(f(\widetilde{Q}_t^\pm)+1)(f(\widetilde{V}_t)+1)}+\frac{\sigma f(\widetilde{Q}_t^\pm)\sqrt{f(\widetilde{V}_t)}\rho}{(f(\widetilde{Q}_t^\pm)+1)(f(\widetilde{V}_t)+1)}\Big) & 0 \\ -\frac{\sigma f(\widetilde{P}^{\pm}_t)\sqrt{f(\widetilde{V}_t)}\sqrt{1-\rho^2}}{(f(\widetilde{P}^{\pm}_t)+1)(f(\widetilde{V}_t)+1)} & +\Big(\frac{\sqrt{f(\widetilde{V}_t)}}{(f(\widetilde{P}^{\pm}_t)+1)(f(\widetilde{V}_t)+1)}-\frac{\sigma f(\widetilde{P}^{\pm}_t)\sqrt{f(\widetilde{V}_t)}\rho}{(f(\widetilde{P}^{\pm}_t)+1)(f(\widetilde{V}_t)+1)}\Big) & 0 \\ \frac{\sigma\sqrt{1-\rho^2}\sqrt{f(\widetilde{V_t})}}{f(\widetilde{V_t})+1} & \frac{\sigma\rho\sqrt{f(\widetilde{V_t})}}{f(\widetilde{V_t})+1} & 0 \end{array}\right],
\end{align}
\begin{align}
\hat{\boldsymbol{\beta}}\left(\boldsymbol{X}_t\right)=\left[\begin{array}{c}
     \mathbf{1}_{\{\widetilde{Q}_t^\pm=0\}}\frac{1}{(f(\widetilde{V}_t)+1)\xi}\\
     \mathbf{1}_{\{\widetilde{P}^{\pm}_t=0\}} \frac{1}{(f(\widetilde{V}_t)+1)\eta}\\
     \frac{\kappa(\theta-f(\widetilde{V_t}))}{f(\widetilde{V_t})+1}-\frac{\sigma^2f(\widetilde{V_t})}{(f(\widetilde{V_t})+1)^2} 
\end{array}\right],
\end{align}
\begin{align}
    \hat{\boldsymbol{\Gamma}}\left(\boldsymbol{X}_t\right)=\left[\begin{array}{ccc} 0 & 0& 0\\
    0 & 0& 0 \\
    \frac{\sigma\sqrt{1-\rho^2}\sqrt{f(\widetilde{V_t})}}{f(\widetilde{V_t})+1} & \frac{\sigma\rho\sqrt{f(\widetilde{V_t})}}{f(\widetilde{V_t})+1} & 0
    \end{array}\right].
\end{align}

Let $\boldsymbol{n}(\boldsymbol{x})$ be the unit normal vector at the boundary point $\boldsymbol{x}$ pointing inwards, which is given by $\boldsymbol{n}(\boldsymbol{x})=\nabla \Phi(\boldsymbol{x})/|\nabla \Phi(\boldsymbol{x})|$. Since $\boldsymbol{\mu}(\boldsymbol{x})$ and $\boldsymbol{\Sigma}(\boldsymbol{x})$ are continuous and bounded; $\hat{\boldsymbol{\beta}}(\boldsymbol{x})\rho(\boldsymbol{x})$ and $\hat{\boldsymbol{\Gamma}}(\boldsymbol{x})\sqrt{\rho(\boldsymbol{x})}$ are continuous and bounded on $\partial\mathbb{S}$; and there exists some constant $C>0$ such that $(\hat{\boldsymbol{\beta}}(\boldsymbol{x})\rho(\boldsymbol{x}))^\top\boldsymbol{n}(\boldsymbol{x})+\rho(\boldsymbol{x})>C$, then there exists at least one solution to SDE system \eqref{eq:final-sticky-process-SDE} as implied by Theorems I.14 in \cite{graham1988martingale}. Consequently the solution to SDE system \eqref{eq:sticky-process-SDE} can be obtained by 
\begin{align}
\begin{cases}
    &D_t^\pm=\exp(-\exp(\widetilde{V}_t)(\exp(\widetilde{Q}_t^\pm)-1)),\\
    &U_t^\pm=\exp(\exp(\widetilde{V}_t)(\exp(\widetilde{P}_t^\pm)-1)), \\
    &V_t=\exp(\widetilde{V}_t)-1.
    \end{cases}
\end{align}
So the Theorem \ref{thm:existence-weak-sol-SDE} is proved.

\end{proof}

\begin{proof}[Proof of the Pricing PDE in Section \ref{sec:pricing-PDE}]

\begin{align}  d\left(e^{-r t} P(t, x, y, z, v;\Phi)\right)   =& -r e^{-r t} P(t, x, y, z, v;\Phi) d t+e^{-r t} d P(t, x, y, z, v;\Phi) \\ = & e^{-r t}\left[-r P d t+P_x d S_t^{ \pm}+P_y d \bar{S}_t^{ \pm}+P_z d \underline{S}_t^{ \pm}+P_t d t+P_v d V_t\right. \\  &\left.+ \frac{1}{2} P_{x x} d\left\langle S^{ \pm}, S^{ \pm}\right \rangle_t+\frac{1}{2} P_{v v} d\left\langle V, V\right \rangle_t+ P_{x v} d\left\langle S^{ \pm}, V\right \rangle_t \right] \\  = & e^{-r t}\left[-r P d t+P_x\left(r x d t+\sqrt{v} x 1_{\{x \neq\{y, z\}\}} d B_t\right)\right. \\ & +P_y\left(r y d t+y d L_t^{1}\left(D^{ \pm}\right)\right)+P_z\left(r z d t-z d L_t^{1}\left(U^{ \pm}\right)\right) \\ & +P_t d t+P_v\left[k_v(\theta-v) d t+\sigma \sqrt{v} d W_t\right] \\ & +\frac{1}{2} P_{x x} v x^2 1_{\{x \neq\{y, z\}\}} d t+\frac{1}{2} P_{v v} \sigma^2 v d t +P_{x v} \sigma v x \rho 1_{\{x \neq\{y, z\}\}} d t \\ = & e^{-r t}\left\{\left[-r P+r x P_x+r y P_y+r z P_z+P_t+k_v(\theta-v) P_v+\frac{1}{2} P_{x x} v x^2\right.\right. \\ & \left.+\frac{1}{2} \sigma^2 v P_{v v}+\sigma v x \rho P_{x v}\right] d t \\ & +\left(P_y y d L_t^{1}\left(D^{ \pm}\right)-\frac{1}{2} P_{x x} v x^2 1_{\{x=y\}} d t-\sigma v x \rho P_{x v} 1_{\{x=y\}} d t\right) \\ & +\left(-P_z z d L_t^{1}\left(U^{ \pm}\right)-\frac{1}{2} P_{x x} v x^2 1\{x=z\} d t -\sigma v x \rho P_{x v} 1_{\{x=z\}} d t\right) \\ & \left.+P_x \sqrt{v} x 1_{\{x \neq\{y, z\}\}} d B_t+P_v \sigma \sqrt{v} d W_t\right\} \\  = & e^{-r t}\left\{\left[-r p+r x P_x+r y P_y+r z P_z+P_t+k_v(\theta-v) P_v+\frac{1}{2} P_{x x} v x^2\right.\right. \\ & \left.+\frac{1}{2} \sigma^2 v P_{v v}+\sigma v x \rho P_{x v}\right] d t \\ & +\left(P_y y d L_t^{1}\left(D^{ \pm}\right)-\frac{1}{2} P_{x x} v y^2 k_s d L_t^{1}\left(D^{ \pm}\right)-\sigma v y \rho P_{x v} k_s d L_t^{1}\left(D^{ \pm}\right)\right) \\ & +\left(-P_z z d L_t^{1}\left(U^{ \pm}\right)-\frac{1}{2} P_{x x} v z^2 \eta d L_t^{1}\left(U^{ \pm}\right)-\sigma v z \rho P_{x v} \eta d L_t^{1}\left(U^{ \pm}\right)\right) \\ & \left.+P_x \sqrt{v} x 1_{\{x \neq\{y, z\}\}} d B_t+P_v \sigma \sqrt{v} d W_t\right\} \\ & \end{align}

For the function $P(t,x,y,z,v;\Phi)\in \mathbf{C}^2$, $e^{-rt}P(t,x,y,z,v;\Phi)$ is a local martingale. Then, for $0\leq t<T$, $0\leq z<x<y$, $P(t,x,y,z,v;\Phi)$ satisfies the PDE
\begin{align}
	\begin{array}{l}
	\frac{1}{2} v x^2 P_{x x}+\rho \sigma v x P_{x v}+\frac{1}{2} \sigma^2 v P_{v v}+r x P_x \\
	+\kappa(\theta-v) P_v+r y P_y+r z P_z+P_t=r P.
	\end{array}
\end{align}
When $x=y$ and $x=z$, the boundary conditions are due to the facts that 
\[
\left(P_y y -\frac{1}{2} P_{x x} v y^2 k_s -\sigma v y \rho P_{x v} k_s \right)d L_t^{1}\left(D^{ \pm}\right)=0
\]
and
\[
\left(-P_z z -\frac{1}{2} P_{x x} v z^2 \eta -\sigma v z \rho P_{x v} \eta \right) d L_t^{1}\left(U^{ \pm}\right)=0.
\]
\end{proof}

\section{Calibration Algorithm}
\label{app-cali}
\begin{algorithm}[H]
  \caption{Deep Levenberg–Marquardt Calibration Algorithm}
  \begin{algorithmic}[1]
      \State Initialize values for parameter set $\phi$, and the Levenberg–Marquardt parameters $\lambda$, $\lambda_{down}$ and $\lambda_{up}$. Obtain the scaling factor $l_d$, underlying asset prices, strike prices and maturities. 
      \State Evaluate the root mean square error, $r\gets \sqrt{\frac{\sum_{i=1}^{N_d}(P^{\theta}_i-P^{MKT}_i)^2}{N_d}}$, and the Jacobian, $J\gets \frac{\partial P^{\theta}}{\partial \phi}$, at the initial parameter guess.
       
      \State $t$ $\gets$ 1, $rmse\gets 0$, $iter\gets 0$
      \While{$t\le 10000$}
        \State $g\gets J^\intercal J+\lambda I$, $\triangledown C\gets J^\intercal r$
        \State $\phi_{new}\gets \phi-g^{-1}\triangledown C$
        \State Check if the parameters in $\phi_{new}$ are within an appropriate range
        \State Evaluate root mean square error $r_{new}$ with $\phi_{new}$
        \If{$r_{new}<r$}
        \State $\phi\gets \phi_{new}$, $r\gets r_{new}$, $\lambda\gets \lambda/\lambda_{down}$, $iter\gets 0$
        \If{$|r_{new}-rmse|<10^{-10}$}
        \State \textbf{break}
        \EndIf
        \Else
        \State $\lambda=\lambda\times \lambda_{up}$
        \State $iter=iter+1$
        \If{$iter>20$}
        \State \textbf{break}
        \EndIf
      \EndIf
      \State $rmse\gets r$
      \State $t\gets t+1$
      \EndWhile

  \end{algorithmic}
\end{algorithm}

\section{Sticky Drawdown Stochastic Volatility model and Sticky Drawup Stochastic Volatility model}
\label{app-SVSDSVSU}

The SVSD model follows the SDE:
\[
\left\{\begin{array}{l}d D_t^{ \pm}=1_{\left\{D_t^{ \pm} \neq 1 \right\}} \sqrt{V_t} D_t^{ \pm} d B_t-D_t^{ \pm} d L_t^1\left(D^{ \pm}\right), \\ d V_t=\kappa\left(\theta-V_t\right) d t+\sigma \sqrt{V_t} d W_t, \\ \int_0^t 1_{\left\{D_s^{ \pm}=1\right\}} d s=\xi L_t^1\left(D^{ \pm}\right).\end{array}\right.
\]

$$
\begin{aligned}
& \left\{\begin{array}{l}
d S_t^{ \pm}=r S_t^{ \pm} d t+\sqrt{V_t} S_t^{ \pm} \mathbf{1}_{\left\{S_t^{ \pm} \neq\bar{S}_t^{ \pm} \right\}} d B_t, \\
d V_t=\kappa\left(\theta-V_t\right) d t+\sigma \sqrt{V_t} d W_t \\
\int_0^t 1_{\left\{S_t^{ \pm}=\bar{S}_t^{ \pm}\right\}} d s=\xi L_t^1\left(D^{ \pm}\right) \\
d \bar{S}_t^{ \pm}=r \bar{S}_t^{ \pm} d t+\bar{S}_t^{ \pm} d L_t^1\left(D^{ \pm}\right) .
\end{array}\right. \\
&
\end{aligned}
$$
By Feynman-Kac Theorem, we have the pricing PDE:
\begin{align}
	\begin{array}{l}
	\frac{1}{2} v x^2 P_{x x}+\rho \sigma v x P_{x v}+\frac{1}{2} \sigma^2 v P_{v v}+r x P_x \\
	+\kappa(\theta-v) P_v+r y P_y+P_t=r P,
	\end{array}
\end{align}
where $0\le t<T$ and $0<x<y$. The boundary is at $x = y$:
\begin{align}
	\begin{array}{l}
	P_y(t,y,y,v;\Phi)=\left[\frac{1}{2} v y P_{x x}(t,y,y,v;\Phi)+\rho \sigma v P_{x v}(t,y,y,v;\Phi)\right] \xi. 
	\end{array}
\end{align}
The terminal condition is
\begin{align}
	P(T, x, y, v;\Phi) = (x-K)^+, 
\end{align}
where $0< x\le y$.

The SVSU model is defined by:
\[
\left\{\begin{array}{l} d U_t^{ \pm}=1_{\left\{U_t^{ \pm} \neq 1\right\}} \sqrt{V_t} U_t^{ \pm} d B_t+U_t^{ \pm} d L_t^1\left(U^{ \pm}\right), \\ d V_t=\kappa\left(\theta-V_t\right) d t+\sigma \sqrt{V_t} d W_t, \\   \int_0^t 1_{\left\{U_s^{ \pm}=1\right\}} d s=\eta L_t^1\left(U^{ \pm}\right) .\end{array}\right.
\]
$$
\begin{aligned}
& \left\{\begin{array}{l}
d S_t^{ \pm}=r S_t^{ \pm} d t+\sqrt{V_t} S_t^{ \pm} \mathbf{1}_{\left\{S_t^{ \pm} \neq\underline{S}_t^{ \pm}\right\}} d B_t, \\
d V_t=\kappa\left(\theta-V_t\right) d t+\sigma \sqrt{V_t} d W_t \\
\int_0^t \mathbf{1}_{\left\{S_t^{ \pm}=\underline{S}_t^{ \pm}\right\}} d s=\eta L_t^1\left(U^{ \pm}\right) \\
d \underline{S}_t^{ \pm}=r \underline{S}_t^{ \pm} d t-\underline{S}_t^{ \pm} d L_t^1\left(U^{ \pm}\right) .
\end{array}\right. \\
&
\end{aligned}
$$
By Feynman-Kac Theorem, we have the pricing PDE: 
\begin{align}
	\begin{array}{l}
	\frac{1}{2} v x^2 P_{x x}+\rho \sigma v x P_{x v}+\frac{1}{2} \sigma^2 v P_{v v}+r x P_x \\
	+\kappa(\theta-v) P_v+r z P_z+P_t=r P,
	\end{array}
\end{align}
where $0\le t<T$ and $0< z<x$. The boundary is at $x=z$:
\begin{align}
	\begin{array}{l}
	P_z(t,z,z,v;\Phi)=-\left[\frac{1}{2} v z P_{x x}(t,z,z,v;\Phi)+\rho \sigma v P_{x v}(t,z,z,v;\Phi)\right] \eta.
	\end{array}
\end{align}
The terminal condition is
\begin{align}
	P(T, x, z, v;\Phi) = (x-K)^+,
\end{align}
where $0< z\le x$.

        \bibliographystyle{chicagoa}

\end{document}